\documentclass[a4paper,12pt,numbers=noenddot]{scrartcl}
\usepackage[left=3.5cm, right=3.5cm, top=2cm, bottom=4cm]{geometry}

\usepackage{caption}
\setkomafont{caption}{\normalfont\upshape}
\setkomafont{captionlabel}{\upshape\bfseries\rmfamily}

\usepackage[T1]{fontenc}
\usepackage[utf8]{inputenc}
\usepackage{lmodern}

\usepackage{graphicx}

\usepackage{mathtools} 
\mathtoolsset{showonlyrefs=true}
\usepackage[ntheorem, framemethod=TikZ]{mdframed} 
\usepackage{framed} 
\usepackage[amsmath, hyperref, framed]{ntheorem}%
\usepackage{amssymb} 
\usepackage{bm}

\newcommand{\nParticles}{N} 
\newcommand{\normConst}{\mathcal{Z}} 

\DeclareMathOperator{\probMeasSet}{\mathcal{P}}%
\DeclareMathOperator{\measSet}{\mathcal{M}}%
\DeclareMathOperator{\boundedFunSet}{\calB}%
\DeclareMathOperator{\osc}{osc}%

\newcommand{\bijection}{\theta}%

\newcommand{\martingaleDiffSingle}{\varDelta U}%
\newcommand{\martingaleSingle}{U}%
\newcommand{\remainderSingle}{L}%
\newcommand{\localErrorSingle}{V}%

\newcommand{\martingale}{\mathcal{U}}%
\newcommand{\remainder}{\mathcal{L}}%
\newcommand{\localError}{\mathcal{V}}%

\newcommand{\martingaleAlt}{\widetilde{\mathcal{U}}}%
\newcommand{\testFun}{f}
\newcommand{\testFunAlt}{g}

\newcommand{\Filt}{\pi}
\newcommand{\Pred}{\tilde{\pi}}
\newcommand{\Trans}{L}
\newcommand{\obs}{g}
\newcommand{\pathObs}{\mathbf{y}} 
\newcommand{\PathObs}{\mathbf{Y}} 
\newcommand{\marginalLikelihood}{\mathcal{L}}
\newcommand{\spaceObs}{F}%
\newcommand{\sigFieldState}{\mathcal{E}}%
\newcommand{\sigFieldObs}{\mathcal{F}}
\newcommand{\sigFieldPath}[1]{\bm{\mathcal{E}}_{\!{#1}}}%

\newcommand{\spaceGeneric}{H}%
\newcommand{\sigFieldGeneric}{\mathcal{H}}%
\newcommand{\spaceGenericAlt}{{H^\prime}}%
\newcommand{\sigFieldGenericAlt}{{\mathcal{H}^\prime}}%

\newcommand{\state}{x} 
\newcommand{\State}{X} 
\newcommand{\path}{\mathbf{x}} 
\newcommand{\Path}{\mathbf{X}}
\newcommand{\stateAlt}{z}

\newcommand{\pathAlt}{\mathbf{z}}

\newcommand{\spaceState}{E} 
\newcommand{\spacePath}{\mathbf{E}}
 
\newcommand{\Target}{\eta} 
\newcommand{\uTarget}{\gamma} 
\newcommand{\ProposalKernel}{R}

\newcommand{\particlePath}[2]{\bm{\xi}_{{#1}}^{#2}}
\newcommand{\ParticlePath}[2]{\bm{\xi}_{{#1}}^{#2}}
\newcommand{\ParticlePathStationary}[2]{\tilde{\bm{\xi}}_{{#1}}^{#2}}

\newcommand{\Mutation}{M}
\newcommand{\Potential}{G}
\newcommand{\PotentialProposal}{F}

\newcommand{\boundResolvent}[1]{\widebar{T}_{#1}}
\newcommand{\boundResolventUniform}{\widebar{T}}
\newcommand{\boundTime}[1]{a_{#1}}
\newcommand{\boundExponent}[1]{b_{#1}}
\newcommand{\boundIntegralOperatorLipschitz}[1]{\widebar{\varGamma}_{#1}}
\newcommand{\boundIntegralOperatorLipschitzAlt}[1]{\widebar{\varUpsilon}_{#1}}

\newcommand{\asymptoticVarianceBPF}[1]{\varSigma_{#1}^{\text{BPF}}}
\newcommand{\asymptoticVarianceMCMCBPF}[1]{\varSigma_{#1}^{\text{MCMC-BPF}}}
\newcommand{\asymptoticVarianceFAAPF}[1]{\varSigma_{#1}^{\text{FA-APF}}}
\newcommand{\asymptoticVarianceMCMCFAAPF}[1]{\varSigma_{#1}^{\text{MCMC-FA-APF}}}

\newcommand{\iact}{\mathrm{iact}}

\newcommand{\mcmcKern}[2]{K_{#1}^{#2}}
\newcommand{\mcmcTarget}[2]{\Phi_{#1}^{#2}}
\newcommand{\mcmcInit}[2]{\kappa_{#1}^{#2}}
\newcommand{\mcmcProposal}[2]{\Pi_{#1}^{#2}}
\newcommand{\integralOperatorLipschitz}[2]{\varGamma_{#1}^{#2}}
\newcommand{\integralOperatorLipschitzAlt}[2]{\varUpsilon_{#1}^{#2}}
\newcommand{\resolv}[2]{T_{#1}^{#2}}
\newcommand{\covarianceFunction}[2]{C_{#1}^{#2}}
\newcommand{\covarianceFunctionAlt}[2]{\mathcal{C}_{#1}^{#2}}

\DeclareMathOperator{\dN}{N} 

\newcommand{\diff}{\mathrm{d}} 
\DeclareMathOperator{\Prob}{\mathbb{P}} 
\DeclareMathOperator{\E}{\mathbb{E}} 
\DeclareMathOperator{\var}{var} 
\DeclareMathOperator{\ind}{\mathbf{1}} 
\newcommand{\ccdot}{\,\cdot\,} 
\newcommand{\ProbSymbol}{\mathbb{P}}%
\DeclareMathOperator{\cov}{cov}%
\DeclareMathOperator{\Id}{Id}%
\newcommand{\unitFun}{\mathbf{1}}%

\newcommand{\convergesInDistribution}{\ensuremath{\mathrel{\mathop{\rightarrow}\nolimits_{\text{\upshape{d}}}}}}
\newcommand{\convergesAlmostSurely}{\ensuremath{\mathrel{\mathop{\rightarrow}\nolimits_{\text{\upshape{a.s.}}}}}}

\newcommand{\EG}{e.g.\@}
\newcommand{\IE}{i.e.\@}

\newcommand{\RHS}{r.h.s.\@}

\newcommand{\calA}{\mathcal{A}}%
\newcommand{\calB}{\mathcal{B}}%
\newcommand{\calF}{\mathcal{F}}%
\newcommand{\calG}{\mathcal{G}}%
\newcommand{\calL}{\mathcal{L}}%
\newcommand{\calR}{\mathcal{R}}%
\newcommand{\reals}{\mathbb{R}}%
\newcommand{\naturals}{\mathbb{N}}%
\DeclareFontFamily{U}{mathx}{\hyphenchar\font45}
\DeclareFontShape{U}{mathx}{m}{n}{<-> mathx10}{}
\DeclareSymbolFont{mathx}{U}{mathx}{m}{n}
\DeclareMathAccent{\widebar}{0}{mathx}{"73}

\usepackage[inline]{enumitem}

\usepackage{hyperref}%
\hypersetup{%
  breaklinks=true,%
  linktocpage=true,
  colorlinks=false,%
  linkcolor=black,%
  citecolor=black,%
  filecolor=black,%
   hidelinks,%
  urlcolor=black,%
  pdftitle={},%
  pdfsubject={},%
  pdfauthor={Axel Finke},%
  pdfkeywords={},%
  pdfproducer={}%
  }%

\usepackage{ifthen}%
\usepackage{mfirstuc}
\usepackage{xkeyval}%
\usepackage{xfor}%
\usepackage{amsgen}%
\usepackage{etoolbox}%
\usepackage{datatool-base}%

\usepackage[%
  xindy,%
  style=long,%
  toc=true,%
  nonumberlist,%
  acronym,%
  acronymlists={ensemble_abbreviations},%
  hyperfirst=true
  ]{glossaries}%

\newacronym{PMCMC}{PMCMC}{particle Markov chain Monte Carlo}%
\newacronym{EMCMC}{EMCMC}{ensemble Markov chain Monte Carlo}%
\newacronym{APF}{APF}{auxiliary particle filter}%
\newacronym{PF}{PF}{particle filter}%
\newacronym{FAAPF}{FA-APF}{fully-adapted auxiliary particle filter}%
\newacronym{BPF}{BPF}{bootstrap particle filter}%
\newacronym{EHMM}{EHMM}{embedded hidden Markov models}%
\newacronym{HMM}{HMM}{hidden Markov model}%
\newacronym{CPF}{CPF}{conditional particle filter}%
\newacronym{CPFBS}{CPF-BS}{conditional particle filter with backward sampling}%
\newacronym{CPFAS}{CPF-AS}{conditional particle filter with backward sampling}%
\newacronym{PG}{PG}{particle Gibbs}%
\newacronym{PGBS}{PG-BS}{particle Gibbs sampler with backward sampling}%
\newacronym{PGAS}{PG-AS}{particle Gibbs sampler with backward sampling}%
\newacronym{SQMC}{SQMC}{sequential quasi Monte Carlo}%
\newacronym{RQMC}{RQMC}{randomised quasi Monte Carlo}%
\newacronym[user1={ancestor-sampling}]{AS}{AS}{ancestor sampling}%
\newacronym[user1={backward-sampling}]{BS}{BS}{backward sampling}%
\newacronym{PDF}{PDF}{probability density function}%
\newacronym{IID}{IID}{independent and identically distributed}%
\newacronym{MCMC}{MCMC}{Markov chain Monte Carlo}%
\newacronym{MH}{MH}{Metropolis--Hastings}%
\newacronym{ESS}{ESS}{effective sample size}%
\newacronym{PMMH}{PMMH}{particle marginal Metro\-po\-lis--Has\-tings}%
\newacronym{MCWM}{MCWM}{Monte Carlo within Metropolis}%
\newacronym{CDF}{CDF}{cumulative distribution function}%
\newacronym{SMC}{SMC}{sequential Monte Carlo}%
\newacronym{CSMC}{CSMC}{conditional sequential Monte Carlo}%
\newacronym[\glslongpluralkey={MCMC particle filters}]{MCMCPF}{MCMC-PF}{MCMC particle filter}%
\newacronym[\glslongpluralkey={MCMC bootstrap particle filters}]{MCMCBPF}{MCMC-BPF}{MCMC bootstrap particle filter}%
\newacronym[\glslongpluralkey={MCMC fully-adapted particle filters}]{MCMCFAAPF}{MCMC-FA-APF}{MCMC fully-adapted auxiliary particle filter}%
\newacronym{MCMCAPF}{MCMC-APF}{MCMC auxiliary particle filters}%
\newacronym{FAAPFalt}{FA-APF}{fully-adapted auxiliary PF}%
\newacronym{BPFalt}{BPF}{bootstrap PF}%

\newacronym{EPSRC}{EPSRC}{Engineering and Physical Sciences Research Council}%
\newacronym{LW}{LW}{Liu~\&~West}%
\newacronym{PL}{PL}{particle learning}%
\newacronym{FF}{FF}{fertility factor}%
\newacronym{CLT}{CLT}{central limit theorem}%
\newacronym{WLLN}{WLLN}{weak law of large numbers}%
\newacronym{SLLN}{SLLN}{strong law of large numbers}%
\newacronym{IACT}{IACT}{integrated autocorrelation time}%
\newacronym{HMC}{HMC}{Hamiltonian Monte Carlo}%
\newacronym{MALA}{MALA}{Metropolis-adjusted Langevin algorithm}%
\usepackage{csquotes}%
\usepackage{natbib}
\setcitestyle{authoryear}
\qedsymbol{\ensuremath{\Box}}
\theoremstyle{plain}%

\theoremheaderfont{\normalfont\bfseries\rmfamily}
\theorembodyfont{\normalfont\slshape}%
\theoremseparator{.}

\newtheorem{proposition}{Proposition}%
\newtheorem{lemma}{Lemma}%
\newtheorem{corollary}{Corollary}%
\newtheorem{example}{Example}%

\theoremstyle{break}
\theorembodyfont{\normalfont}%
\newmdtheoremenv[
 ntheorem=true,
 skipbelow = .6\baselineskip plus 1ex minus 1ex,
 skipabove = .6\baselineskip plus 1ex minus 1ex,
 innerleftmargin = 0pt,
 innerrightmargin = 0pt,
 leftline = false,
 rightline = false,
 needspace = 5ex 
]{framedAlgorithm}[theorem]{Algorithm}

\theoremstyle{empty}%
\theoremheaderfont{}

\theoremstyle{nonumberplain}%
\theorembodyfont{\upshape \rmfamily}%
\theoremheaderfont{\upshape \rmfamily\bfseries}%
\theoremsymbol{\ensuremath{_\Box}}
\newtheorem{proof}{Proof}%

\usepackage[textwidth=2.5cm, textsize=scriptsize]{todonotes}

\usepackage{tikz}
\usetikzlibrary{backgrounds, fit,matrix,positioning, external}

\setkomafont{sectioning}{\normalfont \rmfamily \bfseries}%
\setkomafont{paragraph}{\normalfont \rmfamily \bfseries}%

\title{Limit theorems for\\sequential MCMC methods}

\author{Axel Finke\\Department of Statistics and Applied Probability,\\National University of Singapore, Singapore\\[0.3cm] Arnaud Doucet\\Department of Statistics, Oxford University, UK\\[0.3cm]Adam M.~Johansen\\Department of Statistics, University of Warwick, UK\\and The Alan Turing Institute, London, UK}

\date{\today}

\begin{document}

\maketitle

\enlargethispage{\baselineskip}

\begin{abstract}
  \noindent{}\Gls{SMC} methods, also known as \glspl{PF}, constitute a class of algorithms used to approximate expectations with respect to a sequence of probability distributions as well as the normalising constants of those distributions. Sequential \gls{MCMC} methods are an alternative class of techniques addressing similar problems in which particles are sampled according to an \gls{MCMC} kernel rather than conditionally independently at each time step. These methods were introduced over twenty years ago by \citet{berzuini1997dynamic}. Recently, there has been a renewed interest in such algorithms as they demonstrate an empirical performance superior to that of \gls{SMC} methods in some applications. We establish a \glsdesc{SLLN} and a \glsdesc{CLT} for sequential \gls{MCMC} methods and provide conditions under which errors can be controlled uniformly in time. In the context of state-space models, we provide conditions under which sequential \gls{MCMC} methods can indeed outperform standard \gls{SMC} methods in terms of asymptotic variance of the corresponding Monte Carlo estimators.

\end{abstract}

\section{Introduction}
\glsreset{MCMC}
\glsreset{SMC}
\Gls{SMC} algorithms are used to approximate expectations with respect to a sequence of probability measures as well as the normalizing constants of those measures. These techniques have found numerous applications in statistics, signal processing, physics and related fields \citep[see, \EG,][for a recent review]{kunsch2013particle}. These algorithms proceed in a sequential manner by generating a collection of $N$ conditionally independent particles at each time step. An alternative to these schemes in which the particles at each time step are sampled instead according to a single \gls{MCMC} chain was proposed early on by  \citet{berzuini1997dynamic}. Over recent years, there has been a renewed interest in such ideas as there is empirical evidence that these methods can outperform standard \gls{SMC} algorithms in interesting scenarios \citep[see, \EG,][for novel applications and extensions]{carmi2012gaussian,golightly2006bayesian,septier2009mcmc,septier2016langevin,pal2018sequential}. These methods have been termed \emph{sequential \gls{MCMC}} methods in the literature. However, in this work, we will also refer to these as \emph{\glspl{MCMCPF}}, to convey the idea that they rely on the same importance-sampling construction as particle methods.


Although there is a wealth of theoretical results available for \gls{SMC} algorithms -- see, for example, \citet{delmoral2004feynman} -- to the best of our knowledge, no convergence guarantees have yet been provided for \glspl{MCMCPF}. The present work fills this gap by providing a \glsdesc{SLLN} and a \glsdesc{CLT} for the Monte Carlo estimators of expectations and normalising constants obtained through \glspl{MCMCPF}.
Our results show that compared to conventional \glspl{PF}, the asymptotic variance of estimators obtained by \glspl{MCMCPF} includes additional terms which can be identified as the excess variance arising from the autocorrelation of the \gls{MCMC} chains used to generate the particles. This implies that a standard \gls{PF} always provides estimators with a lower asymptotic variance than the corresponding \gls{MCMCPF} if both algorithms target the same distributions and if the latter relies on positive \gls{MCMC} kernels.

However, \glspl{MCMCPF} exhibit a significant advantage over regular \glspl{PF}. The popular \gls{FAAPF} introduced by \citet{pitt1999filtering} typically significantly outperforms the \gls{BPF} of \citet{gordon1993novel}, for example when approximating the optimal filter for state-space models in the presence of informative measurements. Unfortunately, the \gls{FAAPF} is implementable for only a very restricted class of models whereas the \gls{MCMCPF} version of the \gls{FAAPF} is much more widely applicable. In scenarios in which the \gls{FAAPF} is not implementable, but its \gls{MCMCPF} version is, and in which the \gls{MCMC} kernels used by the latter are sufficiently rapidly mixing, the \gls{MCMCPF} can substantially outperform implementable but rather inefficient standard \glspl{PF} such as the \gls{BPF}.

\section{MCMC-PFs}

\subsection{Notation}

Let $(\Omega, \calA, \ProbSymbol)$ be some probability space and denote expectation with respect to $\ProbSymbol$ by $\E$. For some set measurable space $(\spaceGeneric, \sigFieldGeneric)$, we let $\boundedFunSet(\spaceGeneric)$ denote the Banach space of all of bounded, real-valued, $\sigFieldGeneric$-measurable functions on $\spaceGeneric$, equipped with the uniform norm $\lVert \testFun \rVert \coloneqq \sup_{x \in \spaceGeneric} \lvert \testFun(x)\rvert$. We also endow this space with the Borel $\sigma$-algebra (with respect to $\lVert \ccdot \rVert$), and the product spaces $\boundedFunSet(\spaceGeneric) \times \boundedFunSet(\spaceGeneric)$ and $\boundedFunSet(\spaceGeneric)^d$ for $d \in \naturals$ with the associated product $\sigma$-algebras. We also define the subsets $\boundedFunSet_1(\spaceGeneric) \coloneqq \{\testFun \in \boundedFunSet(\spaceGeneric) \mid \lVert \testFun \rVert \leq 1\}$.
Furthermore, for any $\testFun \in \boundedFunSet(\spaceGeneric)$, $\osc(\testFun) \coloneqq \sup_{(x,y) \in \spaceGeneric^2} \lvert \testFun(x) - \testFun(y) \rvert$. Finally, we let $\unitFun \in \boundedFunSet(\spaceGeneric)$ denote the unit function on $\spaceGeneric$, \IE{} $\unitFun \equiv 1$.

Let $\measSet(\spaceGeneric)$ denote the Banach space of all finite and signed measures on $(\spaceGeneric, \sigFieldGeneric)$ equipped with the total variation norm $\lVert \mu \rVert \coloneqq \frac{1}{2}  \sup_{\testFun \in \boundedFunSet_1(\spaceGeneric)} \lvert \mu(\testFun) \rvert$, where $\mu(\testFun) \coloneqq \int_{\spaceGeneric} \testFun(x) \mu(\diff x)$, for any $\mu \in \measSet(\spaceGeneric)$ and any $\testFun \in \boundedFunSet(\spaceGeneric)$. We define $\probMeasSet(\spaceGeneric) \subseteq \measSet(\spaceGeneric)$ to be the set of all probability measures on $(\spaceGeneric, \sigFieldGeneric)$. 

Let $(\spaceGenericAlt, \sigFieldGenericAlt)$ be another measurable space. For bounded integral operators $M\colon \boundedFunSet(\spaceGenericAlt) \to \boundedFunSet(\spaceGeneric)$, $\testFun \mapsto M(\testFun)(x) \coloneqq \int_\spaceGenericAlt \testFun(z) M(x, \diff z)$ for any $x \in \spaceGeneric$, we define $\smash{[\mu \otimes M](\testFun) =\allowbreak \int_{\spaceGeneric \times \spaceGenericAlt} \mu(\diff x) M(x,\diff y) \testFun(x,y)}$  for any $\mu$ in $\probMeasSet(\spaceGeneric)$ and $\testFun \in \boundedFunSet(\spaceGeneric)\times\boundedFunSet(\spaceGenericAlt)$. We also define the operator norm $\lVert M \rVert \coloneqq
\sup_{\testFun \in \boundedFunSet_1(\spaceGenericAlt)} \lVert M(\testFun)\rVert$ as well as the \emph{Dobrushin coefficient}:
$\beta(M) \coloneqq \sup_{(x,y) \in \spaceGenericAlt \times \spaceGenericAlt} \lVert M(x,\ccdot) - M(y, \ccdot)\rVert.$

Finally, ``$\convergesAlmostSurely$'' denotes almost sure convergence with respect to $\ProbSymbol$ 
and ``$\convergesInDistribution$'' denotes convergence in distribution.

\subsection{Path-space Feynman--Kac model}

We want to approximate expectations under some distributions which are related to a distribution flow $(\Target_n)_{n \geq 1}$ on spaces $(\spacePath_n, \sigFieldPath{n})$ -- with $(\spacePath_1, \sigFieldPath{1}) \coloneqq (\spaceState, \sigFieldState)$ and $(\spacePath_n, \sigFieldPath{n}) \coloneqq (\spacePath_{n-1} \times \spaceState, \sigFieldPath{n-1} \otimes \sigFieldState)$, for $n > 1$ -- of increasing dimension, where
\begin{equation}
 \Target_n(\diff \path_n) \coloneqq \frac{\uTarget_n(\diff \path_n)}{\normConst_n} \in \probMeasSet(\spacePath_n),
\end{equation}
for some positive finite measure $\uTarget_n$ on $(\spacePath_n, \sigFieldPath{n})$ and typically unknown normalising constant $\normConst_n \coloneqq \uTarget_n(\unitFun)$. Throughout this work, we write $\path_p \coloneqq \state_{1:p} = (\path_{p-1}, \state_p)$ and $\pathAlt_p \coloneqq \stateAlt_{1:p} = (\pathAlt_{p-1}, \stateAlt_p)$.

We assume that the target distributions are induced by a Feynman--Kac model on the path space \citep{delmoral2004feynman}. That is, there exists an initial distribution $\Mutation_1 \coloneqq \Target_1 \in \probMeasSet(\spacePath_1)$, a sequence of Markov transition kernels $\Mutation_n\colon \spacePath_{n-1} \times \sigFieldState \to [0,1]$ for $n>1$ and a sequence of bounded (without loss of generality we take the bound to be 1) measurable potential functions $\Potential_n\colon \spacePath_n \to (0,1]$, for $n\geq1$, such that for any $\testFun_n \in \boundedFunSet(\spacePath_n)$,
\begin{align}
 \uTarget_n(\testFun_n) = \Target_1 Q_{1,n}(\testFun_n),
\end{align}
where, hereafter using the convention that any quantity with subscript (\IE{} time index) $0$ is to be ignored from the notation, we have defined the two-parameter semigroup:
\begin{align}
  Q_{p,q}(\testFun_p)(\path_p)
  \coloneqq
  \begin{cases}
   [Q_{p+1} \cdots Q_q](\testFun_q)(\path_p), & \text{if $p < q$,}\\
   \testFun_p(\path_p), & \text{if $p = q$,}
  \end{cases}
\end{align}
for any $1 \leq p \leq q \leq n$, where
\begin{align}
 Q_{p+1}(\path_p, \diff \pathAlt_{p+1})
 \coloneqq \Potential_p(\pathAlt_p) \delta_{\path_p}(\diff \pathAlt_p) \Mutation_{p+1}(\pathAlt_p, \diff \stateAlt_{p+1}).
 \end{align}
This implies that
\begin{align}
 \Target_n(\testFun_n) = \mcmcTarget{n}{\Target_{n-1}}(\testFun_n) \coloneqq \frac{\Target_{n-1} Q_n(\testFun_n)}{\Target_{n-1} Q_n(\unitFun)} = \frac{\uTarget_n(\testFun_n)}{\uTarget_n(\unitFun)},\label{eq:recursioneta}
\end{align}
where we have defined the following family of probability measures
\begin{align}
 \mcmcTarget{n}{\mu}(\diff \path_n)
 & \coloneqq
 \begin{dcases}
  \Mutation_1(\diff \path_1) = \Target_1(\diff \path_1), & \text{if $n = 1$,}\\
  \frac{\Potential_{n-1}(\path_{n-1})}{\mu(\Potential_{n-1})}[\mu \otimes \Mutation_n](\diff \path_n)
  , & \text{if $n > 1$,}
 \end{dcases}
\end{align}
indexed by $\mu \in \probMeasSet(\spacePath_{n-1})$.

For later use, we also define the family of normalised operators
\begin{align}
 \widebar{Q}_{p,n}(\testFun_n)(\path_{p})
 & \coloneqq \frac{Q_{p,n}(\testFun_n)(\path_{p})}{\Target_p Q_{p,n}(\unitFun)}.
\end{align}
Note that this implies that $\Target_n(\testFun_n) = \Target_p \widebar{Q}_{p,n}(\testFun_n)$ for any $1 \leq p \leq n$ and any $\testFun_n \in \boundedFunSet(\spacePath_n)$.

\subsection{Generic MCMC-PF algorithm}
\glsreset{MH}
In Algorithm \ref{alg:mcmc_pf}, we summarise a generic \gls{MCMCPF} scheme for constructing sampling approximations $\Target_n^N$ of $\Target_n$. It admits all the \glspl{MCMCPF} discussed in this work as special cases. We recall that by convention, any quantity with subscript $0$ is to be ignored from the notation.  This algorithm is essentially a \gls{PF} in which the particles are not sampled conditionally independently from $\mcmcTarget{n}{\mu}$ at step~$n$, for some $\mu \in \probMeasSet(\spacePath_{n-1})$,  but are instead sampled according to a Markov chain with initial distribution $\mcmcInit{n}{\mu}\in \probMeasSet(\spacePath_{n})$ and Markov transition kernels $\mcmcKern{n}{\mu}\colon \spacePath_{n} \times \sigFieldPath{n} \to [0,1]$ which are invariant with respect to $\mcmcTarget{n}{\mu}$.

\noindent\parbox{\textwidth}{
\begin{flushleft}
  \begin{framedAlgorithm}[generic \gls{MCMCPF}] \label{alg:mcmc_pf}
At time~$1$,
 \begin{enumerate}
  \item sample $\ParticlePath{1}{1} \sim \smash{\mcmcInit{1}{\Target_{0}^N}}$ and $\ParticlePath{1}{i} \sim \smash{\mcmcKern{1}{\Target_{0}^N}}(\particlePath{1}{i-1}, \ccdot)$, for $2 \leq i \leq N$,
  \item set $\Target_1^N \coloneqq \frac{1}{N} \sum_{i=1}^N \delta_{\ParticlePath{1}{i}}$.
 \end{enumerate}
 At time~$n$, $n > 1$,
 \begin{enumerate}
  \item sample $\ParticlePath{n}{1} \sim \smash{\mcmcInit{n}{\Target_{n-1}^N}}$ and $\ParticlePath{n}{i} \sim \smash{\mcmcKern{n}{\Target_{n-1}^N}}(\particlePath{n}{i-1}, \ccdot)$, for $2 \leq i \leq N$,
  \item set $\Target_n^N \coloneqq \frac{1}{N} \sum_{i=1}^N \delta_{\ParticlePath{n}{i}}$.
 \end{enumerate}
\end{framedAlgorithm}
\end{flushleft}
}

For any  time $n \geq 1$ and for any $\testFun_n \in \boundedFunSet(\spacePath_n)$, $\smash{\uTarget_n^N(\testFun_n) \coloneqq \Target_n^N(\testFun_n) \prod_{p=1}^{n-1} \Target_p^N(\Potential_p)}$ is an estimate of $\uTarget_n(\testFun_n)$. In particular, an estimate of the normalising constant $\normConst_n$ is given by
\begin{align}
 \normConst_n^N \coloneqq \uTarget_n^N(\unitFun) = \prod_{p=1}^{n-1} \Target_p^N(\Potential_p) =  \prod_{p=1}^{n-1} \frac{1}{N} \sum_{i=1}^N \Potential_p(\ParticlePath{p}{i}). \label{eq:normalising_constant_estimate}
\end{align}

We hereafter write $\mcmcTarget{n}{N} \coloneqq \mcmcTarget{n}{\Target_{n-1}^N}$, $\mcmcInit{n}{N} \coloneqq \mcmcInit{n}{\Target_{n-1}^N}$ and $\mcmcKern{n}{N} \coloneqq \mcmcKern{n}{\Target_{n-1}^N}$ to simplify the notation. Note that standard \glspl{PF} are a special case of Algorithm~\ref{alg:mcmc_pf} corresponding to $\mcmcKern{n}{N}(\path_{n}, \ccdot) \equiv \mcmcTarget{n}{N}(\ccdot)  = \mcmcInit{n}{N}( \ccdot)$. Unfortunately, implementing standard \glspl{PF} can become prohibitively costly whenever there is no cheap way of generating $N$ \gls{IID} samples from $\mcmcTarget{n}{N}$ -- which can be the case when $\mcmcTarget{n}{N}$ is chosen for reasons of statistical efficiency rather than computational convenience, as in the case of the \gls{FAAPF} of \citet{pitt1999filtering}. In contrast, Algorithm~\ref{alg:mcmc_pf} only requires the construction of \gls{MCMC} kernels which leave this distribution invariant.

Practitioners typically initialise the Markov chains close to stationarity by selecting $\mcmcInit{n}{N}=[\Target_{n-1}^N \otimes \Mutation_n'] (\mcmcKern{n}{N})^{N_{\mathrm{burnin}} }$ for some approximation $\Mutation_n'(\path_{n-1}, \diff \state_n)$ of $\Mutation_n(\path_{n-1}, \diff \state_n)$ \citep[see, \EG,][]{carmi2012gaussian,golightly2006bayesian,septier2009mcmc, septier2016langevin}. Here, $N_{\mathrm{burnin}} \geq 1$ denotes a suitably large number of iterations whose samples are discarded as ``burn-in''.  Proposition \ref{prop:slln}, below, will demonstrate that, under regularity conditions, such algorithms can provide strongly consistent estimates of quantities of interest in spite of this out-of-equilibrium initialisation.

In situations in which it is possible to initialise the Markov chains at stationarity, \IE{} in which we can initialise $\smash{\ParticlePath{n}{1} \sim \mcmcInit{n}{N}=\mcmcTarget{n}{N}}$, \citet{finke2016embedded} showed that the estimator $\normConst_n^N$ given in \eqref{eq:normalising_constant_estimate} is unbiased as for standard \glspl{PF} \citep{delmoral2004feynman}. This remarkable unbiasedness property permits the use of \glspl{MCMCPF} within pseudo-marginal algorithms \citep{andrieu2009pseudo} and thus to perform Bayesian parameter inference for state-space models. As such an initialisation only requires \emph{one} draw from $\smash{\mcmcTarget{n}{N}}$, the use of relatively expensive methods, such as rejection sampling, may be justifiable. This is in contrast to standard \glspl{PF} which require $N$ such draws the cost of which may be prohibitive.

Furthermore, the conditional \gls{SMC} scheme proposed in \citet{andrieu2010particle} can also be extended to \glspl{MCMCPF} as demonstrated in \citet{shestopaloff2018sampling}. In this case, construction of a suitable initial distribution $\smash{\mcmcInit{n}{N}}$ is not needed.


The literature on \gls{MCMC} algorithms provides numerous ways in which to construct the Markov kernels  $\mcmcKern{n}{\mu}$. For instance, we could use \gls{MH} \citep{berzuini1997dynamic}, \glsdesc{MALA}, \glsdesc{HMC} and hybrid kernels \citep{septier2009mcmc, septier2016langevin}, kernels based on invertible particle flow ideas \citep{li2017sequential} or on the bouncy particle sampler \citep{pal2018sequential}. As an illustration, Example~\ref{ex:independent_mh} describes a simple independent \gls{MH} kernel with a proposal distribution tailored to our setting.

\begin{example}[independent MH] \label{ex:independent_mh}
  For any $n \geq 1$ and any $\mu \in \probMeasSet(\spacePath_{n-1})$, define the proposal distribution
  \begin{align}
  \mcmcProposal{n}{\mu}(\diff \path_n)
  & \coloneqq \label{eq:indpendent_mh_proposal_distribution}
  \begin{dcases}
    \ProposalKernel_1(\diff \path_1), & \text{if $n = 1$,}\\
    \frac{\PotentialProposal_{n-1}(\path_{n-1})}{\mu(\PotentialProposal_{n-1})}[\mu \otimes \ProposalKernel_n](\diff \path_n), & \text{if $n > 1$,}
  \end{dcases}
  \end{align}
  for some sequence of non-negative bounded measurable functions $\PotentialProposal_n\colon \spacePath_n \to [0,\infty)$, some distribution $\ProposalKernel_1 \in \probMeasSet(\spacePath_1)$ with $\Mutation_1 \ll \ProposalKernel_1$, and some sequence of Markov transition kernels $\ProposalKernel_n\colon \spacePath_{n-1} \times \sigFieldState \to [0,1]$ with $\Mutation_n(\path_{n-1}, \ccdot) \ll \ProposalKernel_n(\path_{n-1}, \ccdot)$, for any $\path_{n-1} \in \spacePath_{n-1}$; for any $n \geq 1$, both $\PotentialProposal_{n-1}$ and $\ProposalKernel_n$ are assumed to satisfy
  \begin{align}
   \sup_{\path_{n} \in \spacePath_{n}} \frac{\Potential_{n-1}(\path_{n-1})}{ \PotentialProposal_{n-1}(\path_{n-1})} \frac{\diff \Mutation_n(\path_{n-1}, \ccdot)}{\diff \ProposalKernel_n(\path_{n-1}, \ccdot)}(\state_n) < \infty. \label{eq:bounded_radon-nikodym_derivatives}
  \end{align}
The independent \gls{MH} kernel $\mcmcKern{n}{\mu}$ with proposal distribution $\mcmcProposal{n}{\mu}$ and target\slash invariant distribution $\mcmcTarget{n}{\mu}$ is given by
  \begin{align}
    \mcmcKern{n}{\mu}(\path_n, \diff \pathAlt_n)
    & \coloneqq \alpha_n(\path_n, \pathAlt_n) \mcmcProposal{n}{\mu}(\diff \pathAlt_n)\\
    & \quad + \biggl(1 - \int_{\spacePath_n}\alpha_n(\path_n, \diff \path_n') \mcmcProposal{n}{\mu}(\diff \path_n')\biggr) \delta_{\path_n}(\diff \pathAlt_n),
  \end{align}
  with acceptance probability
  \begin{align}
    \alpha_n(\path_n, \pathAlt_n) \label{eq:independent_mh_acceptance_probability}
    & \coloneqq 1 \wedge \dfrac{\diff \mcmcTarget{n}{\mu}}{\diff \mcmcProposal{n}{\mu}}(\pathAlt_n) \bigg/ \dfrac{\diff \mcmcTarget{n}{\mu}}{\diff \mcmcProposal{n}{\mu}}(\path_n)\\
    & = 1 \wedge \dfrac{\dfrac{\Potential_{n-1}}{\PotentialProposal_{n-1}}(\pathAlt_{n-1})}{\dfrac{\Potential_{n-1}}{\PotentialProposal_{n-1}}(\path_{n-1})} \dfrac{\dfrac{\diff \Mutation_n(\pathAlt_{n-1}, \ccdot)}{\diff \ProposalKernel_n(\pathAlt_{n-1}, \ccdot)}(\stateAlt_n)}{\dfrac{\diff \Mutation_n(\path_{n-1}, \ccdot)}{\diff \ProposalKernel_n(\path_{n-1}, \ccdot)}(\state_n)}.
  \end{align}
 This acceptance probability notably does not depend on $\mu$.
\end{example}

\subsection{Computational cost}
\label{subsec:computational_cost}

If we are interested only in approximating the normalising constant $\normConst_n$ and if $\Potential_{n-1}(\path_{n-1})$ and $\Mutation_n(\path_{n-1}, \ccdot)$ depend upon only a fixed number of the most recent component(s) of $\path_{n-1}$ (as is the case in the state-space models discussed below), Algorithm~\ref{alg:mcmc_pf} can be implemented at a per-time-step complexity (in both space and time) that is linear in the number of particles $N$ and constant in the time horizon~$n$.

\subsection{Application to state-space models}
\label{subsec:application_to_state-space_models}

Let $(\spaceObs, \sigFieldObs)$ be another measurable space. The \gls{MCMCPF} may be used for (but is not limited to) performing inference in a state-space model given by the bivariate Markov chain $(X_n, Y_n)_{n \geq 1}$ on $(\spaceState \times \spaceObs, \sigFieldState \vee \sigFieldObs)$ with initial distribution $\Trans_1(\diff x_1) \obs_1(x_1, y_1)\psi(\diff y_1)$ and with Markov transition kernels (for any $n > 1$)
\begin{equation}
 \Trans_{n}(x_{n-1}, \diff x_{n})\obs_{n}(x_{n}, y_{n}) \psi(\diff y_n).
\end{equation}
Here, $\Trans_1 \in \probMeasSet(\spaceState)$ is some initial distribution for $X_1$, $\Trans_{n} : \spaceState \times \sigFieldState \to [0,1]$, for $n > 1$, is a Markov transition kernel. 
Furthermore, $\psi$ is some suitable $\sigma$-finite dominating measure on $(\spaceObs, \sigFieldObs)$ and some positive bounded function $\obs_n(\ccdot, y_n)$ so that $\obs_n(x_n, y_n) \psi(\diff y_n)$ represents the transition kernels for the observation at time~$n$. Usually, we can only observe realisations of $(Y_n)_{n \geq 1}$ whereas the process $(X_n)_{n \geq 1}$ is latent.

Assume that we have observed realisations $\pathObs_n = (y_1, \dotsc, y_n)$ of $\PathObs_n \coloneqq (Y_1, \dotsc, Y_n)$, then we often wish to compute (expectations under) the
\begin{itemize}
 \item \emph{filter:} $\Filt_n(\testFun_n) \coloneqq \E[\testFun_n(\Path_n)|\PathObs_n = \pathObs_n]$, for $\testFun_n \in \boundedFunSet(\spacePath_n)$,
 \item \emph{predictor:} $\Pred_n(\testFun_n) \coloneqq \E[\testFun_n(\Path_n)|\PathObs_{n-1} = \pathObs_{n-1}]$, for $\testFun_n \in \boundedFunSet(\spacePath_n)$,
 \item \emph{marginal likelihood:} $\marginalLikelihood_n \coloneqq \E[\prod_{p=1}^{n} \obs_p(X_p, y_p)]$.
\end{itemize}
Note that the definitions of ``filter'' and ``predictor'' here refer to the historical process as we are taking a path-space approach. These terms are sometimes reserved for the final-component marginals of $\Filt_n$ and $\Pred_n$; we will use the terms \emph{marginal filter} and \emph{marginal predictor} for those objects.

\begin{example}[\gls{BPF}-type flow] \label{ex:bpf_flow}
  If for any $n \geq 1$,
  \begin{align}
    \Potential_n(\path_n)
    &\coloneqq \obs_n(x_n, y_n),\\
    \Mutation_{n}(\path_{n-1}, \diff x_n)
    &\coloneqq \Trans_{n}(x_{n-1}, \diff x_n), \label{eq:bpf_mutation}
  \end{align}
  then $\Target_n = \Pred_n$ is the time-$n$ predictor, we can recover the time-$n$ filter as $\Target_n(\Potential_n\testFun_n)/\Target_n(\Potential_n) = \Filt_n(\testFun_n)$, and $\normConst_{n+1} = \marginalLikelihood_{n}$ is the marginal likelihood associated with the observations $\pathObs_{n}$ (with $\normConst_1 = 1$).

  In this case, Algorithm~\ref{alg:mcmc_pf} can be implemented (\EG{} using the independent \gls{MH} kernel from Example~\ref{ex:independent_mh}) as long as $\obs_n$, $\PotentialProposal_{n}$ and $\diff \Trans_{n}(x_{n-1}, \ccdot) / \diff \ProposalKernel_n(\path_{n-1},\ccdot)$ can be evaluated point-wise.
\end{example}

\begin{example}[\gls{FAAPF}-type flow] \label{ex:fa-apf_flow}
  If for any $n \geq 1$,
  \begin{align}
   \Potential_n(\path_n)
   &\coloneqq \Trans_{n+1}(\obs_{n+1}(\ccdot, y_{n+1}))(x_n), \label{eq:fa-apf_potential}\\
   \Mutation_{n}(\path_{n-1}, \diff x_{n})
   &\coloneqq \frac{\Trans_{n}(x_{n-1}, \diff x_{n})\obs_{n}(x_{n},y_{n})}{\Trans_{n}(\obs_{n}(\ccdot, y_{n}))(x_{n-1})},  \label{eq:fa-apf_mutation}
  \end{align}
  then $\Target_n = \Filt_n$ is the time-$n$ filter, we can recover the time-$n$ predictor as $\Target_{n-1}\otimes \Trans_n = \Pred_n$, and $\normConst_n = \marginalLikelihood_{n}$ is the marginal likelihood associated with the observations $\pathObs_{n}$.

  For this flow, it follows from \eqref{eq:recursioneta} that sampling $\ParticlePath{n}{i}$ from $\mcmcTarget{n}{N}$ requires first sampling an index $J = j \in \{1,...,N\}$ with probability proportional to $\Potential_{n-1}(\particlePath{n-1}{j})$, setting the first $n-1$ components of $\smash{\ParticlePath{n}{i}}$ equal to $\smash{\particlePath{n-1}{j}}$ and then sampling the final component according to $\smash{\Mutation_{n}(\particlePath{n-1}{j}, \ccdot)}$. There are many scenarios in which this is not feasible as both \eqref{eq:fa-apf_potential} and \eqref{eq:fa-apf_mutation} involve an intractable integral.
  However, designing an \gls{MCMC} kernel of invariant distribution $\mcmcTarget{n}{N}$ is a much easier task as the product $\Potential_{n-1}(\path_{n-1})\Mutation_{n}(\path_{n-1}, \diff x_{n})$ does not involve any intractable integral. For example, if we use the independent \gls{MH} kernel from Example~\ref{ex:independent_mh} then the acceptance probability in \eqref{eq:independent_mh_acceptance_probability} reduces to (for simplicity, we take $\PotentialProposal_{n-1} \equiv 1$):
  \begin{align}
    \alpha_n(\path_n, \pathAlt_n)
    & = 1 \wedge \dfrac{\obs_{n}(\stateAlt_n, y_n)}{\obs_{n}(\state_{n}, y_n)} \dfrac{\dfrac{\diff \Trans_n(\stateAlt_{n-1}, \ccdot)}{\diff \ProposalKernel_n(\pathAlt_{n-1}, \ccdot)}(\stateAlt_n)}{\dfrac{\diff \Trans_n(\state_{n-1}, \ccdot)}{\diff \ProposalKernel_n(\path_{n-1}, \ccdot)}(\state_n)}.
  \end{align}
\end{example}

\begin{example}[general \gls{APF}-type flow] \label{ex:general_apf_flow}
  Let $\Target_1$ be some approximation of $\Filt_1$ and, for $n \geq 1$, let $\Mutation_{n+1}(\path_{n}, \diff x_{n+1})$ be some approximation of \eqref{eq:fa-apf_mutation} as well as
  \begin{align}
   \Potential_n(\path_n)
   &\coloneqq
   \begin{dcases}
    \frac{\diff \Trans_1}{\diff \Target_1}(\state_1) \obs_1(x_1, y_1) \tilde{\obs}_1(x_1, y_2), & \text{if $n = 1$,}\\
    \frac{\diff \Trans_n(\state_{n-1}, \ccdot)}{\diff \Mutation_n(\path_{n-1}, \ccdot)}(\state_n) \frac{\obs_n(x_n, y_n) \tilde{\obs}_n(x_n, y_{n+1})}{\tilde{\obs}_{n-1}(x_{n-1}, y_n)}, & \text{if $n > 1$.}
   \end{dcases}
  \end{align}
  Here, $\tilde{\obs}_n(x_n, y_{n+1})$ denotes some tractable approximation of \eqref{eq:fa-apf_potential} which can be evaluated point-wise. More generally, we could incorporate information from observations at times $n+1,\dotsc,n+l$ for some $l \geq 1$ into $\Mutation_{n+1}(\path_{n}, \diff x_{n+1})$ and replace $\tilde{\obs}_n(x_n, y_{n+1})$ by some approximation of $\smash{\int_{\spaceState^l} \prod_{p=n+1}^{n+l} \Trans_{p}(\path_{p-1}, \diff \state_p) \obs_p(\state_p, y_p)}$ as in the case of \emph{lookahead} methods \citep[see][for example]{lin2013}.

  Note that the (general) \gls{APF} flow admits the two other flows as special cases. That is, taking $\Mutation_n$ as in \eqref{eq:bpf_mutation} and $\tilde{\obs}_n \equiv 1$ yields \gls{BPF}-type flow; taking $\Mutation_n$ as in \eqref{eq:fa-apf_mutation} and $\tilde{\obs}_n(x_n, y_{n+1}) = \Trans_{n+1}(\obs_{n+1}(\ccdot, y_{n+1}))(x_n)$ yields the \gls{FAAPF}-type flow.
\end{example}

In the remainder of this work, we will refer to Algorithm~\ref{alg:mcmc_pf} as the \emph{\gls{MCMCBPF}} whenever the distribution flow $(\Target_n)_{n \geq 1}$ is defined as in Example~\ref{ex:bpf_flow}, as \emph{\gls{MCMCFAAPF}} whenever the flow is defined as in Example~\ref{ex:fa-apf_flow} and as \emph{\gls{MCMCAPF}} whenever the flow is defined as in Example~\ref{ex:general_apf_flow}. Furthermore, we drop the prefix ``\gls{MCMC}'' when referring to the conventional \gls{PF}-analogues of these algorithms, \IE{} in the case that $\smash{\mcmcKern{n}{\mu}(\path_n,\ccdot) \equiv \mcmcTarget{n}{\mu} = \mcmcInit{n}{\mu}}$.

\section{Main Results}
\glsreset{SLLN}
\glsreset{CLT}
In this section, we state a \gls{SLLN} (Proposition~\ref{prop:slln}) and a \gls{CLT} (Proposition~\ref{prop:clt}) for the approximations of the normalised and unnormalised flows $(\Target_n)_{n \geq 1}$ and $(\uTarget_n)_{n \geq 1}$ generated by an \gls{MCMCPF}.

\subsection{Assumptions}

We make the following assumptions about the \gls{MCMC} kernels used to sample the particles at each time step. The first assumption on the \gls{MCMC} kernels ensures that they are suitably ergodic (it corresponds to assuming that the kernels used are uniformly ergodic, uniformly in their invariant distribution) and is the only assumption required to obtain the \gls{SLLN}.
The second assumption on the \gls{MCMC} kernels is a Lipschitz-type condition.

\begin{enumerate}[label=\textbf{(A\arabic*)}, ref=\textbf{A\arabic*}]
 \item \label{as:ergodicity}
  For any $n \geq 1$, there exists $i_n \in \naturals$ and $\varepsilon_n(K) \in (0,1]$  such that for all $\mu \in \probMeasSet(\spacePath_{n-1})$ and all $\path_n, \pathAlt_n \in \spacePath_n$:
  \begin{align}
   (\mcmcKern{n}{\mu})^{i_n}(\path_n, \ccdot) \geq \varepsilon_n(K) (\mcmcKern{n}{\mu})^{i_n}(\pathAlt_n, \ccdot).
  \end{align}

 \item \label{as:lipschitz}
 For any $n\geq 1$, there exists a constant $\boundIntegralOperatorLipschitz{n} < \infty$ and a family of bounded integral operators $(\integralOperatorLipschitz{n}{\mu})_{\mu \in \probMeasSet(\spacePath_{n-1})}$ from $\boundedFunSet(\spacePath_{n-1})$ to $\boundedFunSet(\spacePath_n)$ such that for any $(\mu, \nu) \in \probMeasSet(\spacePath_{n-1})^2$ and any $\testFun_n \in \boundedFunSet(\spacePath_n)$,
  \begin{align}
    \lVert [\mcmcKern{n}{\mu} - \mcmcKern{n}{\nu}](\testFun_n) \rVert
    &\leq \int_{\boundedFunSet(\spacePath_{n-1})} \lvert [\mu - \nu](\testFunAlt) \rvert \integralOperatorLipschitz{n}{\mu}(\testFun_n, \diff \testFunAlt) \label{eq:lipschitz:1}
  \end{align}
  and 
  \begin{align}
    \int_{\boundedFunSet(\spacePath_{n-1})} \lVert \testFunAlt \rVert \integralOperatorLipschitz{n}{\mu}(\testFun_n, \diff \testFunAlt) \leq \lVert \testFun_n \rVert \boundIntegralOperatorLipschitz{n}. \label{eq:lipschitz:2}
  \end{align}
\end{enumerate}
Recall that for any bounded integral operator $M\colon \boundedFunSet(\spaceGeneric) \to \boundedFunSet(\spaceGeneric)$, $\beta(M) \coloneqq \sup_{(x,y) \in \spaceGeneric \times \spaceGeneric} \lVert M(x,\ccdot) - M(y, \ccdot)\rVert$ is the associated Dobrushin coefficient. Note that Assumption~\ref{as:ergodicity} implies that
\begin{equation}
  \sup_{\mu \in \probMeasSet(\spacePath_{n-1})}\beta((\mcmcKern{n}{\mu})^{i_n}) \leq 1 - \varepsilon_n(K) < 1. \label{eq:as:ergodicity}
\end{equation}
In addition, if $(\ParticlePath{p}{i})_{i \geq 1}$ and $(\ParticlePathStationary{p}{i})_{i \geq 1}$ are Markov chains with transition kernels $\mcmcKern{n}{\mu}$, with $\smash{(\ParticlePathStationary{p}{i})_{i \geq 1}}$ initialised from stationarity, then a standard coupling argument shows that Assumption~\ref{as:ergodicity} also implies that for any $\nParticles, r \in \naturals$ and any $\testFun_n \in \boundedFunSet(\spacePath_n)$ with $\lVert \testFun_n\rVert \leq 1$,
\begin{align}
 \E\biggl[\biggl\lvert \sum_{i=1}^\nParticles \testFun_n(\ParticlePath{n}{i}) - \testFun_n(\ParticlePathStationary{n}{i}) \biggr\rvert^r\biggr]^{\mathrlap{\frac{1}{r}}}
 & \leq \sum_{i=1}^\nParticles \E\bigl[\bigl\lvert \testFun_n(\ParticlePath{n}{i}) - \testFun_n(\ParticlePathStationary{n}{i}) \bigr\rvert^r\bigr]^{\frac{1}{r}}\\
 & \leq 2 \lVert \testFun_n \rVert \sum_{i=1}^\nParticles (1 - \varepsilon_n(K))^{\lfloor i/i_n \rfloor}\\
 & \leq 2 i_n / \varepsilon_n(K) \eqqcolon \boundResolvent{n}. \label{eq:bound_resolvent}
\end{align}

The assumptions are similar to those imposed in \citet{bercu2012fluctuations}. They are strong and rarely hold for non-compact spaces. It might be possible to adopt weaker conditions such as those in \citet{andrieu2011nonlinear} but this would involve substantially more technical and complicated proofs. As an illustration, we show that Assumptions~\ref{as:ergodicity} and \ref{as:lipschitz} hold if we employ the independent \gls{MH} kernels from Example~\ref{ex:independent_mh}, at least if $\spaceState$ is finite.

\begin{example}[independent MH, continued]
 Assumption~\ref{as:ergodicity} is satisfied due to \citet[][Theorem~2.1]{mengersen1996rates}. To see this, note that by \eqref{eq:bounded_radon-nikodym_derivatives}, for any $n \geq 1$ and any $\mu \in \probMeasSet(\spacePath_{n-1})$, and since $\PotentialProposal_{n-1}$ is bounded and $\Potential_{n-1} > 0$,
 \begin{align}
  \sup_{\path_n \in \spacePath_n}\frac{\diff \mcmcTarget{n}{\mu}}{\diff \mcmcProposal{n}{\mu}}(\path_n)
  \leq \frac{\lVert\PotentialProposal_{n-1}\rVert}{\mu(\Potential_{n-1})} \sup_{\path_{n} \in \spacePath_{n}} \frac{\Potential_{n-1}(\path_{n-1})}{ \PotentialProposal_{n-1}(\path_{n-1})} \frac{\diff \Mutation_n(\path_{n-1}, \ccdot)}{\diff \ProposalKernel_n(\path_{n-1}, \ccdot)}(\state_n) < \infty.
 \end{align}
 Assumption~\ref{as:lipschitz} was proved for finite spaces $\spaceState$ (and in the case $\PotentialProposal_{n} = \Potential_n$ but the extension to $\PotentialProposal_n \neq \Potential_n$ is immediate) in \citet[][Section~2]{bercu2012fluctuations}.
\end{example}

When proving time-uniform convergence results, we also make the following assumptions on the mutation kernels and potential functions of the Feynman--Kac model. The first of these ensures that Assumptions~\ref{as:ergodicity} holds uniformly in time. The second and third of these constitute strong mixing conditions that have been extensively used in the analysis of \gls{SMC} algorithms, although they can often be relaxed in similar settings this comes at the cost of greatly complicating the analysis \citep[see, \EG,][]{whiteley2013stability, douc2014long}.

\begin{enumerate}[label=\textbf{(B\arabic*)}, ref=\textbf{B\arabic*}]

 \item \label{as:stability_mcmc_kernels} $\bar{\imath} \coloneqq \sup_{n \geq 1} i_n < \infty$ and $\varepsilon(K) \coloneqq \inf_{n \geq 1} \varepsilon_n(K) > 0$.


 \item \label{as:stability_mutation} There exists ${m} \in \naturals$ and $\varepsilon(\Mutation) \in (0,1]$ such that for any $n \geq 1$, any $\path_n, \pathAlt_n \in \spacePath_n$ and any $\varphi \in \boundedFunSet(\spaceState)$:
 \begin{align}
  \MoveEqLeft \int_{\spaceState^m} \biggl[\prod_{\smash{p=1}}^{m} \Mutation_{n+p}(\path_{n+p-1}, \diff \state_{n+p})\biggr]\varphi(\state_{n+{m}})\\
  & \geq \varepsilon(\Mutation) \int_{\spaceState^{m}} \biggl[\prod_{p=1}^{\smash{{m}}} \Mutation_{n+p}(\pathAlt_{n+p-1}, \diff \stateAlt_{n+p})\biggr]\varphi(\stateAlt_{n+{m}}).
 \end{align}


 \item \label{as:stability_potential} There exists $l \in \naturals$ and $\varepsilon(\Potential) \in (0,1]$ such that for any $n \geq 1$ and any $\path_n, \pathAlt_n \in \spacePath_n$:
 \begin{align}
  \Potential_n(\path_n) = \Potential_n((\pathAlt_{n-l-1}, \state_{((n-l) {\vee} 1):n})) \quad \text{and} \quad
  \Potential_n(\path_n) \geq \varepsilon(\Potential) \Potential_n(\pathAlt_n).
 \end{align}

\end{enumerate}
Under these conditions, time-uniform bounds will be obtained when the test function under study has supremum norm of at most $1$ and depends upon only its final coordinate marginal, \IE{} we will restrict our attention to test functions $\testFun_n \in \smash{\boundedFunSet_1^\star}(\spacePath_n)^d$, where $\smash{\boundedFunSet_1^\star(\spacePath_n) \coloneqq \{ \testFun_n' \in \boundedFunSet^\star(\spacePath_n) \,|\,\lVert\testFun_n'\rVert \leq 1\}}$ with
\begin{align}
  \smash{\boundedFunSet^\star}(\spacePath_n)
  \coloneqq \{ \testFun_n' \in \boundedFunSet(\spacePath_n) \,|\, \exists \, \varphi \in \boundedFunSet(\spaceState):  \testFun_n' = \varphi \mathrel{\circ} \zeta_{n}\}.
 \end{align}
Here, for any $n \geq 1$, $\zeta_{n}: \spacePath_n \to \spaceState$ denotes the canonical final-coordinate projection operator defined through $\path_n \mapsto \zeta_{n}(\path_n) \coloneqq \state_n$. In the state space context this corresponds, essentially, to considering the approximation of the marginal filter and predictor rather than their path-space analogues.

\subsection{Strong law of large numbers}
The first main result in this work is the \gls{SLLN} stated in Proposition \ref{prop:slln}. Its proof is an immediate consequence of a slightly stronger $\mathbb{L}_r$-inequality given in Proposition~\ref{prop:lr_inequality}, the proof of which will be given in Appendix~\ref{app:proofs}.

\begin{proposition}[$\mathbb{L}_r$-inequality]~\label{prop:lr_inequality}
 Under Assumption~\ref{as:ergodicity}, for each $r,n \geq 1$ there exist $\boundTime{n}, \boundExponent{r} < \infty$ such that for any $\testFun_n \in \boundedFunSet(\spacePath_n)$ and any $N \geq 1$:
 \begin{align}
  \E\bigl[ \bigl\lvert [\Target_n^N - \Target_n](\testFun_n) \bigr\rvert^r \bigr]^{\frac{1}{r}} \leq \frac{\boundTime{n} \boundExponent{r}}{\sqrt{N}}\lVert \testFun_n \rVert. \label{eq:lr-error}
 \end{align}
  Under the additional Assumptions~\ref{as:stability_mcmc_kernels}--\ref{as:stability_potential} and if $\testFun_n \in \smash{\boundedFunSet_1^\star}(\spacePath_n)$, the \RHS{} of \eqref{eq:lr-error} is bounded uniformly in time, \IE{} there exist $a < \infty $ such that $\sup_{n \geq 1 }a_n \leq a$.
\end{proposition}

\begin{proposition}[strong law of large numbers]~\label{prop:slln}
 Under Assumption~\ref{as:ergodicity}, for any $n,d \geq 1$ and $\testFun_n \in \boundedFunSet(\spacePath_n)^d$, as $N\to \infty$,
 \begin{enumerate}
  \item \label{prop:slln:normalised_predictors} $\Target_n^N(\testFun_n) \convergesAlmostSurely \Target_n(\testFun_n)$,
  \item \label{prop:slln:unnormalised_predictors} $\uTarget_n^N(\testFun_n) \convergesAlmostSurely \uTarget_n(\testFun_n)$.
 \end{enumerate}
\end{proposition}
\begin{proof}
 Without loss of generality, we prove the result for scalar-valued test functions $\testFun_n \in \boundedFunSet(\spacePath_n)$.  Part~\ref{prop:slln:normalised_predictors} is a direct consequence of Proposition~\ref{prop:lr_inequality}, for some $r > 2$, using the Borel--Cantelli Lemma together with Markov's inequality.  Part~\ref{prop:slln:unnormalised_predictors} follows from Part~\ref{prop:slln:normalised_predictors} and boundedness of the potential functions $\Potential_p$, \IE
 \begin{align}
  \uTarget_n^N(\testFun_n) = \Target_n^N(\testFun_n) \prod_{p=1}^{n-1} \Target_p^N(\Potential_p) \convergesAlmostSurely \Target_n(\testFun_n) \prod_{p=1}^{n-1} \Target_p(\Potential_p) = \uTarget_n(\testFun_n),
 \end{align}
 as $N \to \infty$. This completes the proof. \hfill \ensuremath{_\Box}
\end{proof}

\subsection{Central limit theorem}

The second main result is Proposition~\ref{prop:clt} which adapts the usual \gls{CLT} for \gls{SMC} algorithms \citep[\EG][Propositions~9.4.1 \& 9.4.2]{delmoral2004feynman} to our setting. Its proof is given in Appendix~\ref{app:proofs}. As in \citet{delmoral2010interacting, bercu2012fluctuations}, we will make extensive use of the resolvent operators $\resolv{n}{\mu}$, for any $\mu \in \probMeasSet(\spacePath_{n-1})$ and any $\testFun_n \in \boundedFunSet(\spacePath_n)$ defined by
\begin{equation}
 \resolv{n}{\mu}(\testFun_n) \coloneqq \sum_{j = 0}^\infty [(\mcmcKern{n}{\mu})^j - \mcmcTarget{n}{\mu}](\testFun_n).
\end{equation}
These operators satisfy the \emph{Poisson equation}
\begin{align}
 (\mcmcKern{n}{\mu} - \Id)\resolv{n}{\mu} &= \mcmcTarget{n}{\mu} - \Id, \label{eq:poisson_equation_1}\\
 \mcmcTarget{n}{\mu} \resolv{n}{\mu} &\equiv 0. \label{eq:poisson_equation_2}
\end{align}
 Under Assumption~\ref{as:ergodicity}, \citet[Proposition~3.1]{bercu2012fluctuations} show that 
  \begin{equation}
  \sup_{\mu \in \probMeasSet(\spacePath_{n-1})}
  \lVert \resolv{n}{\mu}\rVert \leq \boundResolvent{n},\label{eq:bound_on_resolvent}
\end{equation}
where $\boundResolvent{n}$ is given in \eqref{eq:bound_resolvent}.

In the following, for $n\geq 1$, we consider a vector-valued test function $\testFun_n = (\testFun_n^u)_{1 \leq u \leq d} \in \boundedFunSet(\spacePath_n)^d$. Using the resolvent operators, for any $1 \leq u,v \leq d$, we define the covariance function $\covarianceFunction{n}{\mu}(\testFun_n^u,\testFun_n^v)\colon \spacePath_n \to \reals$, for any $\path_n \in \spacePath_n$ given by
\begin{align}
 \MoveEqLeft\covarianceFunction{n}{\mu}(\testFun_n^u,\testFun_n^v)(\path_n) \label{eq:definition_of_covFun}\\
 & \coloneqq \mcmcKern{n}{\mu}[(\resolv{n}{\mu}(\testFun_n^u)-\mcmcKern{n}{\mu}\resolv{n}{\mu}(\testFun_n^u)(\path_n))(\resolv{n}{\mu}(\testFun_n^v)-\mcmcKern{n}{\mu}\resolv{n}{\mu}(\testFun_n^v)(\path_n))](\path_n).
\end{align}
Under Assumption~\ref{as:ergodicity}, we have $\covarianceFunction{n}{\mu}(\testFun_n^u,\testFun_n^v) \in \boundedFunSet(\spacePath_n)$, for any $1 \leq u, v, \leq d$ and any $\mu \in \probMeasSet(\spacePath_{n-1})$. Indeed, using \eqref{eq:bound_on_resolvent}, it is straightforward to check that
\begin{equation}
  \sup_{\mu \in \probMeasSet(\spacePath_{n-1})}\lVert \covarianceFunction{n}{\mu}(\testFun_n^u, \testFun_n^v) \rVert \leq 4 \boundResolvent{n}^2 \lVert \testFun_n^u\rVert \lVert \testFun_n^v \rVert. \label{eq:boundedness_of_covFun}
\end{equation}

Throughout the remainder of this work, let $\localErrorSingle = (\localErrorSingle_n)_{n \geq 1}$ be a sequence of independent and centred Gaussian fields with
\begin{align}
 \E[\localErrorSingle_n(\testFun_n^u)\localErrorSingle_n(\testFun_n^v)] = \Target_n\covarianceFunction{n}{\Target_{n-1}}(\testFun_n^u, \testFun_n^v), \label{eq:local_error_covariance_function}
\end{align}
and define the $(d,d)$-matrix $\varSigma_n(\testFun_n) \coloneqq (\varSigma_n(\testFun_n^u, \testFun_n^v))_{1\leq u,v\leq d}$ by
\begin{equation}
 \varSigma_n(\testFun_n^u, \testFun_n^v)
 \coloneqq \sum_{p=1}^n
 \E[\localErrorSingle_n(\widebar{Q}_{p,n}(\testFun_n^u))\localErrorSingle_n(\widebar{Q}_{p,n}(\testFun_n^v))],
 \label{eq:asymptotic_variance}
\end{equation}
for any $n \geq 1$, any $\testFun_n = (\testFun_n^u)_{1 \leq u \leq d} \in \boundedFunSet(\spacePath_n)^d$ and any $1 \leq u,v \leq d$. Additionally, let $\dN(0, \varSigma)$ denote a (multivariate) centred Gaussian distribution with some covariance matrix $\varSigma$.

\begin{proposition}[central limit theorem]~\label{prop:clt}
Under Assumptions~\ref{as:ergodicity} and \ref{as:lipschitz}, for any $n,d \geq 1$ and any $\testFun_n \in \boundedFunSet(\spacePath_n)^d$, as $N \to \infty$,
\begin{align}
  \label{enum:prop:clt:unnormalised} \frac{\sqrt{N}}{\uTarget_n(\unitFun)} [\uTarget_n^N - \uTarget_n](\testFun_n)
  & \convergesInDistribution \sum_{p=1}^n \localErrorSingle_p(\widebar{Q}_{p,n}(\testFun_n)) \sim \dN(0, \varSigma_n(\testFun_n)),
\end{align}
and likewise, writing $\bar{\testFun}_n \coloneqq \testFun_n - \Target_n(\testFun_n)$,
\begin{align}
 \label{enum:prop:clt:normalised} \sqrt{N}[\Target_n^N - \Target_n](\testFun_n)
 & \convergesInDistribution \sum_{p=1}^n  \localErrorSingle_p(\widebar{Q}_{p,n}(\bar{\testFun}_n)) \sim \dN(0, \varSigma_n(\bar{\testFun}_n)).
\end{align}
Under the additional Assumptions~\ref{as:stability_mcmc_kernels}--\ref{as:stability_potential} and if $\testFun_n \in \smash{\boundedFunSet_1^\star}(\spacePath_n)^d$, the asymptotic variance in \eqref{enum:prop:clt:normalised} is bounded uniformly in time, \IE{} there exists $c < \infty$ such that
\begin{align}
 \sup \varSigma_n(\bar{\testFun}_n^u, \bar{\testFun}_n^v) \leq c,
 \label{eq:bound_on_asymptotic_variance}
\end{align}
where the supremum is over all $n \geq 1$, $\testFun_n \in \smash{\boundedFunSet_1^\star}(\spacePath_n)^d$ and $1 \leq u,v \leq d$.

\end{proposition}

\section{Comparison with standard PFs}

\subsection{Variance decomposition}

In this section, we first examine the asymptotic variance from Proposition~\ref{prop:clt}. We then illustrate the trade-off between \glspl{MCMCPF} and standard \glspl{PF}.

To ease the exposition, we only consider scalar-valued test functions $\testFun_n \in \boundedFunSet(\spacePath_n)$ throughout this section. As noted in \citet[Proposition~3.6]{bercu2012fluctuations}, the terms $\mcmcTarget{n}{\mu} \covarianceFunction{n}{\mu}(\testFun_n, \testFun_n)$ from \eqref{eq:local_error_covariance_function} which, via \eqref{eq:asymptotic_variance}, appear in the expressions for the asymptotic variance in Proposition~\ref{prop:clt} can be written in the following form which is more commonly used in the \gls{MCMC} literature:
 \begin{align}
  \mcmcTarget{n}{\mu} \covarianceFunction{n}{\mu}(\testFun_n, \testFun_n)
  & = \textstyle \int_{\spacePath_n^2} \mcmcTarget{n}{\mu}(\diff \path_n)\mcmcKern{n}{\mu}(\path_n, \diff \pathAlt_n)\bigl[\resolv{n}{\mu}(\testFun_n)(\pathAlt_n) - \mcmcKern{n}{\mu} \resolv{n}{\mu}(\testFun_n)(\path_n)\bigr]^2\\
  & = \mcmcTarget{n}{\mu}\bigl(\resolv{n}{\mu}(\testFun_n)^2 - \mcmcKern{n}{\mu}\resolv{n}{\mu}(\testFun_n)^2\bigr)\\
  & = \mcmcTarget{n}{\mu}\bigl(\resolv{n}{\mu}(\testFun_n)^2 - [\mcmcTarget{n}{\mu}(\testFun_n) - \testFun_n + \resolv{n}{\mu}(\testFun_n)]^2\bigr) \quad \text{[by \eqref{eq:poisson_equation_1}]}\\
  & = \mcmcTarget{n}{\mu}\bigl(- [\testFun_n - \mcmcTarget{n}{\mu}(\testFun_n)]^2 - 2[\mcmcTarget{n}{\mu}(\testFun_n) - \testFun_n]\resolv{n}{\mu}(\testFun_n)\bigr)\\
  & = \mcmcTarget{n}{\mu}\bigl(- [\testFun_n - \mcmcTarget{n}{\mu}(\testFun_n)]^2 + 2 \testFun_n\resolv{n}{\mu}(\testFun_n)\bigr) \quad \text{[by \eqref{eq:poisson_equation_2}]}\\
  &= \var_{\mcmcTarget{n}{\mu}}[\testFun_n] \times \iact_{\mcmcKern{n}{\mu}}[\testFun_n]. \label{eq:asymptotic_variance_decomposition}
 \end{align}
 Here, for any probability measure $\nu \in \probMeasSet(\spacePath_n)$ and any $\nu$-invariant Markov kernel $K$, we have defined the \gls{IACT}:
 \begin{align}
  \iact_{K}[\testFun_n]
  & \coloneqq 1 + 2 \sum_{j=1}^\infty \frac{\cov_{\nu}[\testFun_n, K^j(\testFun_n)]}{\var_{\nu}[\testFun_n]},
 \end{align}
where $\cov_{\nu}[\testFun_n, \testFunAlt_n] \coloneqq \nu([\testFun_n - \nu(\testFun_n)][\testFunAlt_n - \nu(\testFunAlt_n)])$ and $\var_{\nu}[\testFun_n] \coloneqq \cov_{\nu}[\testFun_n, \testFun_n] = \nu([\testFun_n - \nu(\testFun_n)]^2)$.

If the \gls{MCMC} kernels $\mcmcKern{n}{\mu}$ are perfectly mixing, that is if $\mcmcKern{n}{\mu}(\path_n, \ccdot) = \mcmcTarget{n}{\mu}(\ccdot)$ for all $\path_n \in \spacePath_n$, then $\iact_{\mcmcKern{n}{\mu}}[\testFun_n] = 1$, \IE{} $\mcmcTarget{n}{\mu} \covarianceFunction{n}{\mu}(\testFun_n, \testFun_n) = \var_{\mcmcTarget{n}{\mu}}[\testFun_n]$, and the expressions for the asymptotic variances in Proposition~\ref{prop:clt} (as specified through \eqref{eq:local_error_covariance_function} and \eqref{eq:asymptotic_variance}) simplify to those obtained in \citet{chopin2004central, delmoral2004feynman, kunsch2005recursive} for conventional \gls{SMC} algorithms. Thus, by the decomposition from \eqref{eq:asymptotic_variance_decomposition}, the terms appearing in the asymptotic variance of the \gls{MCMCPF} are equal to those appearing in the asymptotic variance of standard \glspl{PF} multiplied by the \gls{IACT} associated with the \gls{MCMC} kernels used to generate the particles.

For positive \gls{MCMC} operators, the \gls{IACT} terms are greater than $1$ for any $\testFun_n \in \boundedFunSet(\spacePath_n)$ and represent the variance ``penalty'' incurred due to the additional between-particle positive correlations in \glspl{MCMCPF} relative to standard \glspl{PF}. Examples of positive operators include the independent \gls{MH} kernel \citep{liu1996metropolized} discussed in Example \ref{ex:independent_mh}, the \gls{MH} kernel with Gaussian or Student-t random walk proposals \citep{baxendale2005renewal} or autoregressive positively correlated proposals with normal or Student-t innovations \citep{doucet2015efficient} as well as some versions of the hit-and-run and slice sampling algorithms \citep{rudolf2013positivity}.

\subsection{Variance--variance trade-off}
\label{subsec:variance-variance_trade_off}

There is an efficiency trade-off involved in deciding whether to employ a standard \gls{PF} or an \gls{MCMCPF} for a particular application. For the same distribution flow $(\Target_n)_{n \geq 1}$ the former always has a lower asymptotic variance than the latter if the \gls{MCMC} draws are positively correlated. However, as we seek to illustrate in the remainder of this section, an \gls{MCMCPF} may still be preferable (in terms of asymptotic variance) to a standard \gls{PF} in certain situations, even if a positive MCMC kernel is used, because it can sometimes be used to target a more efficient distribution flow, \IE{} a flow for which the variance terms $\var_{\mcmcTarget{n}{\mu}}[\testFun_n]$ are reduced far enough to compensate for the \gls{IACT}-based ``penalty'' terms $\iact_{\mcmcKern{n}{\mu}}[\testFun_n]$ in \eqref{eq:asymptotic_variance_decomposition}. Additionally the computational cost of generating one particle in an \gls{MCMCPF} can be smaller than the corresponding cost in a standard \gls{PF}.

As an illustration, we compare the asymptotic variances of approximations $\Filt_n^N$ of the filter $\Filt_n$ either computed using the standard \glspl{PF} or \glspl{MCMCPF} targeting the \gls{BPF} and \gls{FAAPF} flows in the state-space model from Subsection~\ref{subsec:application_to_state-space_models}. 
In the remainder of this section, we let $S_{p,n}\colon \boundedFunSet(\spacePath_p) \to \boundedFunSet(\spacePath_n)$ be a kernel that satisfies $\Filt_p S_{p,n} = \Filt_n$ and which is given by
\begin{align}
 S_{p,n}(\testFun_n)(\path_p)
 & \coloneqq \frac{\marginalLikelihood_p}{\marginalLikelihood_n} \int_{\spacePath_n}  \testFun_n(\pathAlt_n) \delta_{\path_p}(\diff \pathAlt_p) \smashoperator{\prod_{q=p+1}^n} \obs_{q}(\stateAlt_q, y_q) \Trans_{q}(\stateAlt_{q-1}, \diff \stateAlt_q).
\end{align}
We begin by deriving expressions for the asymptotic variances in each case.

\begin{itemize}
  \item \textbf{\gls{BPF} flow.} For the \gls{BPF} flow from Example~\ref{ex:bpf_flow},
  expectations under the filter $\Filt_n(\testFun_n) = \Target_n(\Potential_n\testFun_n)/\Target_n(\Potential_n)$ may then be approximated by $\Filt_n^N(\testFun_n) = \Target_n^N(\Potential_n\testFun_n)/\Target_n^N(\Potential_n)$. Accounting for this transformation \citep[\EG{} as in][]{johansen2007auxiliary} yields
  \begin{align}
   [\Filt_n^N - \Filt_n](\testFun_n)
   & = \frac{\Target_n(\Potential_n)}{\Target_n^N(\Potential_n)} \frac{\Target_n^N(\Potential_n[\testFun_n - \Filt_n(\testFun_n)])}{\Target_n(\Potential_n)}\\
   & = \frac{\Target_n(\Potential_n)}{\Target_n^N(\Potential_n)} [\Target_n^N - \Target_n](\Potential_n[\testFun_n - \Filt_n(\testFun_n)]/\Target_n(\Potential_n)).
  \end{align}
  As Proposition~\ref{prop:slln} ensures that $\Target_n(\Potential_n)/\Target_n^N(\Potential_n) \convergesAlmostSurely 1$, Slutzky's Lemma and Proposition~\ref{prop:clt} are sufficient to show that for the \gls{BPF} and \gls{MCMCBPF}, respectively, $\sqrt{N}[\Filt_n^N - \Filt_n](\testFun_n)$ converges in distribution to a Gaussian distribution with zero mean and variance
  \begin{align}
   \asymptoticVarianceBPF{n}(\testFun_n)
   & = \sum_{p=1}^n \var_{\Pred_p}[\tilde{\testFun}_{p,n}],\label{eq:asymptoticVarianceBpf}\\
   \asymptoticVarianceMCMCBPF{n}(\testFun_n)
   & = \sum_{p=1}^n \var_{\Pred_p}[\tilde{\testFun}_{p,n}] \times \iact_{\mcmcKern{p}{{\Pred_{p-1}}}}[\tilde{\testFun}_{p,n}],\label{eq:asymptoticVarianceMcmcBpf}
  \end{align}
  with, using that $\Target_n(\Potential_n [\testFun_n - \Filt_n(\testFun_n)]/\Target_n(\Potential_n)) = 0$,
  \begin{align}
   \tilde{\testFun}_{p,n}(\path_p)
   & \coloneqq \widebar{Q}_{p,n}(\Potential_n[\testFun_n - \Filt_n(\testFun_n)]/\Target_n(\Potential_n))(\path_p)\\
   & = \obs_p(\state_p, y_{p}) \frac{\marginalLikelihood_{p-1}}{\marginalLikelihood_p} S_{p,n}(\testFun_n - \Filt_n(\testFun_n))(\path_p).
  \end{align}

\item \textbf{\gls{FAAPF} flow.} For the \gls{FAAPF} flow from Example~\ref{ex:fa-apf_flow}, $\Filt_n = \Target_n$,  so that we may approximate the filter by $\Filt_n^N \coloneqq \Target_n^N$. Hence, Proposition~\ref{prop:clt} shows that for the \gls{FAAPF} and \gls{MCMCFAAPF}, respectively, $\sqrt{N}[\Filt_n^N - \Filt_n](\testFun_n)$ converges in distribution to a Gaussian distribution with zero mean and variance
  \begin{align}
   \asymptoticVarianceFAAPF{n}(\testFun_n)
   & = \sum_{p=1}^n \var_{\Filt_p}[\testFun_{p,n}], \label{eq:asymptoticVarianceFaapf}\\
   \asymptoticVarianceMCMCFAAPF{n}(\testFun_n)
   & = \sum_{p=1}^n \var_{\Filt_p}[\testFun_{p,n}] \times \iact_{\mcmcKern{p}{{\Filt_{p-1}}}}[\testFun_{p,n}],\label{eq:asymptoticVarianceMcmcFaapf}
  \end{align}
  with
  \begin{align}
  \testFun_{p,n}(\path_p)
   \coloneqq \widebar{Q}_{p,n}(\testFun_n - \Filt_n(\testFun_n))(\path_p)
   = S_{p,n}(\testFun_n - \Filt_n(\testFun_n))(\path_p).
  \end{align}
\end{itemize}

For the remainder of this section, assume that the asymptotic variance of the standard \gls{FAAPF} is lower than that of the standard \gls{BPF} for the given state-space model. More precisely, we assume that for each $p \leq n$
\begin{align}
 \var_{\Filt_p}[\testFun_{p,n}] \leq \var_{\Pred_p}[\tilde{\testFun}_{p,n}], \quad \textrm{ and hence that } \quad \asymptoticVarianceFAAPF{n}(\testFun_n) \leq \asymptoticVarianceBPF{n}(\testFun_n).
\end{align}
This is thought to hold in many applications and has been empirically verified \EG{} in \citet{snyder2015performance}, although it is  possible to construct counter-examples \citep{johansen2008note}. Assuming that the \gls{MCMC} kernels $\mcmcKern{p}{\mu}$ are positive operators, then the \glspl{IACT} take values in $[1,\infty)$ and hence
\begin{align}
 \asymptoticVarianceBPF{n}(\testFun_n) \leq \asymptoticVarianceMCMCBPF{n}(\testFun_n) \quad \text{and} \quad \asymptoticVarianceFAAPF{n}(\testFun_n) \leq \asymptoticVarianceMCMCFAAPF{n}(\testFun_n).
\end{align}
However, as noted in Example~\ref{ex:fa-apf_flow}, there are many scenarios where \gls{FAAPF} cannot be implemented as we cannot generate $N$ (conditionally) \gls{IID} samples from $\mcmcTarget{n}{N}$. In this case, practitioners typically have to resort to using the standard \gls{BPF} instead. In contrast, the \gls{MCMCFAAPF} can usually be implemented. In such circumstances, use of \glspl{MCMCPF} (specifically in the form of the \gls{MCMCFAAPF}) can preferable, \EG{} if the variance reductions attained by targeting the \gls{FAAPF} flow are large enough to outweigh the additional variance due to the increased particle correlation, \IE{} if for each $1 \leq p \leq n$,
\begin{align}
 \var_{\Filt_p}[\testFun_{p,n}] \times \iact_{\mcmcKern{p}{{\Filt_{p-1}}}}[\testFun_{p,n}] \leq \var_{\Pred_p}[\tilde{\testFun}_{p,n}]
\end{align}
because then
\begin{align}
 \asymptoticVarianceMCMCFAAPF{n}(\testFun_n) \leq \asymptoticVarianceBPF{n}(\testFun_n).
\end{align}

\subsection{Numerical illustration}

We end this section by illustrating the `variance--variance trade-off' mentioned above on two instances of the state-space model from Subsection~\ref{subsec:application_to_state-space_models}.

The first model is a state-space model on a binary space $\spaceState = \spaceObs \coloneqq \{0,1\}$ and with $n=2$ observations: $y_1 = y_2 = 0$. Furthermore, for some $\alpha, \varepsilon \in [0,1]$ and for any $x_1, x_2 \in \spaceState$, $\mu \in \probMeasSet(\spacePath_{n-1})$ and any $n \in \{1,2\}$,
\begin{gather}
 \Trans_1(\{x_1\}) \coloneqq 1/2, \quad \Trans_2(x_1, \{x_2\}) \coloneqq \alpha \ind\{x_2 = x_1\} + (1-\alpha) \ind\{x_2 \neq x_1\},\\
 \obs_n(x_n, y_n) \coloneqq 0.99 \ind\{y_n = x_n\} + 0.01 \ind\{y_n \neq x_n\},\\
 \mcmcKern{n}{\mu}(\path_{n}, \ccdot) \coloneqq \varepsilon \delta_{\path_{n}} + (1-\varepsilon) \mcmcTarget{n}{\mu}.
\end{gather}
While this is clearly only a toy model, we consider it for two reasons. Firstly, it allows us to analytically evaluate the asymptotic variances for standard \glspl{PF} and \glspl{MCMCPF} given in \eqref{eq:asymptoticVarianceBpf}, \eqref{eq:asymptoticVarianceMcmcBpf}, \eqref{eq:asymptoticVarianceFaapf} and \eqref{eq:asymptoticVarianceMcmcFaapf}. Secondly, as discussed in \citet{johansen2008note}, the model allows us to select the parameter $\alpha$ in such a way that the \gls{FAAPF} has either a lower or higher asymptotic variance than the \gls{BPF}.

Figure~\ref{fig:binary_ssm} displays the asymptotic variances relative to the asymptotic variance of the standard \gls{BPF} for the test Function $\testFun_2(\path_{2}) = x_2$ and for two different values of the parameter $\alpha$. As displayed in the first panel, a relatively large value of $\alpha$ leads to the somewhat contrived case that the \gls{BPF} is more efficient than the \gls{FAAPF}. However, as displayed in the second panel, a small value of $\alpha$ makes the \gls{FAAPF} more efficient than the \gls{BPF}. This is because if the system is in state $0$ at time $1$, the time-$2$ proposal used by the \gls{FAAPF} incorporates the observation $y_2 = 0$ and whereas the time-$2$ proposal used by the \gls{BPF} almost always proposes moves to state $1$. In this case, the \gls{MCMCFAAPF} then outperforms the \gls{BPF} as long as the autocorrelation of the \gls{MCMC} kernels used by the former, $\varepsilon$, is not too large.

\begin{figure}[h]
 \centering
 \includegraphics[scale=1.2, trim = 2cm 1.8cm 1cm 1.5cm]{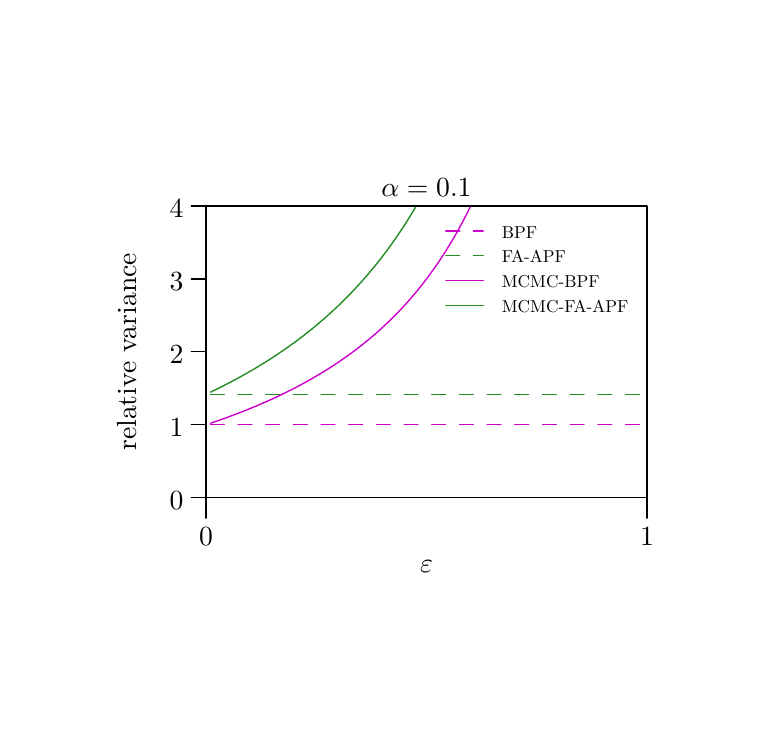}
 \includegraphics[scale=1.2, trim = 1.8cm 1.8cm 1.8cm 1.5cm]{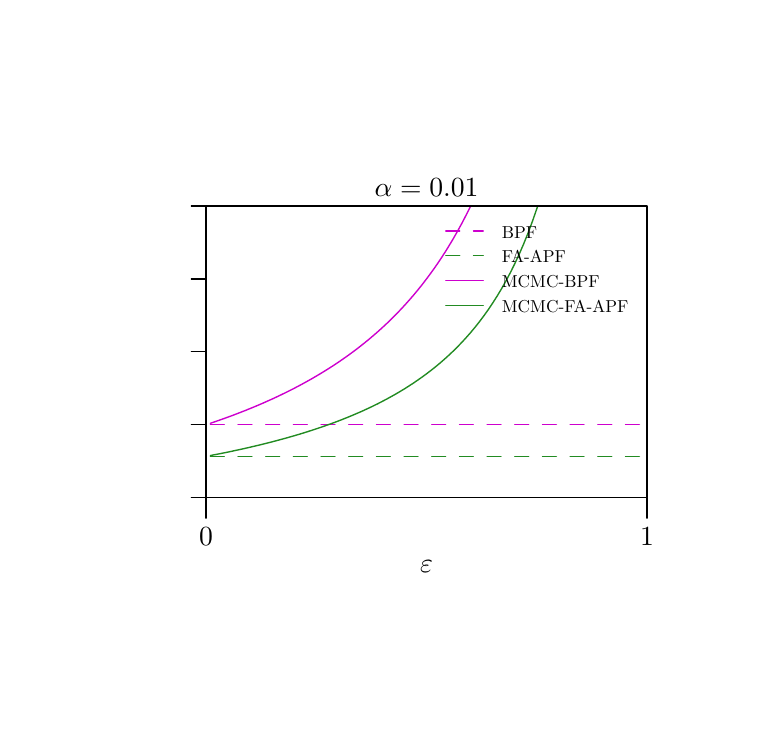}
 \caption{Asymptotic variances (relative to the asymptotic variance of the \gls{BPF}) of the algorithms discussed in Subsection~\ref{subsec:variance-variance_trade_off} in the case that the \gls{BPF} flow is more efficient than the \gls{FAAPF} flow (first panel) and in the case that the \gls{BPF} flow is less efficient than the \gls{FAAPF} flow (second panel).}
 \label{fig:binary_ssm}
\end{figure}


We stress that in practical situations, one might expect a much more pronounced difference between the performance of the \gls{FAAPF} and the \gls{BPF} than is observed in this toy model, and hence that Markov kernels with rather modest mixing properties can still give rise to an \gls{MCMCFAAPF} that can outperform the \gls{BPF} in some situations. Indeed, this appears to be the case in the second model discussed below.

The second model is a $d$-dimensional linear Gaussian state-space model given by $\spaceState = \spaceObs \coloneqq \reals^d$. Furthermore, writing the $d$-dimensional state and observation vector at time $n$ as $\State_n = (\State_{n,i})_{1 \leq i \leq d}$ and $Y_n = (Y_{n,i})_{1 \leq i \leq d}$, respectively,
\begin{gather}
 \frac{\diff \Trans_1}{\diff \lambda^{\otimes d}}(\state_1) = \prod_{i=1}^d \phi(\state_{1,i}), \quad  \frac{\diff \Trans_n(\state_{n-1}, \ccdot)}{\diff \lambda^{\otimes d}}(\state_n) = \prod_{i=1}^{\smash{d}} \phi(\state_{n,i} - \state_{n-1,i}/2),\\
 \obs_n(x_n, y_n) = \prod_{i=1}^d \phi(\state_{n,i} - y_{n,i}),
\end{gather}
where $\lambda$ denotes the Lebesgue measure on $\reals$ and $\phi$ denotes a Lebesgue-density of a univariate standard normal distribution. We take $\mcmcKern{n}{\mu}(\path_{n}, \ccdot)$ to be a \gls{MH} kernel with proposal
\begin{align}
  \mcmcProposal{n}{\mu}(\path_n, \diff \pathAlt_n)
  & =
  \begin{dcases}
    \ProposalKernel(\state_1, \diff \stateAlt_1), & \text{if $n = 1$,}\\
    \frac{\PotentialProposal_{n-1}(\pathAlt_{n-1})}{\mu(\PotentialProposal_{n-1})} \mu(\diff \pathAlt_{n-1}) \ProposalKernel(\state_n, \diff \stateAlt_n), & \text{if $n > 1$,}
  \end{dcases}
\end{align}
where $\ProposalKernel$ is a Gaussian random-walk kernel on $\spaceState$ with transition density
\begin{align}
  \frac{\diff R(\state_{n}, \ccdot)}{\diff \lambda^{\otimes d}}(\stateAlt_n) = \prod_{i=1}^d \sqrt{d} \phi(\sqrt{d}[\stateAlt_{n,i} - \state_{n,i}]),
\end{align}
and where $\PotentialProposal_n(\path_n) = \obs_n(\state_n, y_n)$ for the \gls{MCMCBPF} as well as $\PotentialProposal_n \equiv 1$ for the \gls{MCMCFAAPF}.

For the \gls{MCMCBPF}, the \gls{MCMC} chains at each time step are initialised from stationarity, \IE{}
\begin{align}
 \kappa_n^\mu(\diff \path_n) = \mcmcTarget{n}{\mu}(\diff \path_n) = \dfrac{\mu(\diff \path_{n-1})\obs_{n-1}(\state_{n-1}, y_{n-1})}{\mu(\obs_{n-1}(\ccdot, y_{n-1}))} \Trans_n(\state_{n-1}, \diff \state_n),
\end{align}
as this is almost always possible, in practice. For the \gls{MCMCFAAPF}, the \gls{MCMC} chains are initialised by discarding the first $N_{\mathrm{burnin}} = 100$ samples as burn-in, \IE{}  $\kappa_n^\mu = [\mu \mathbin{\otimes} \Trans_n](\mcmcKern{n}{\mu})^{N_{\mathrm{burnin}}}$.

Figure~\ref{fig:linear_ssm} displays estimates of the marginal likelihood relative to the true marginal likelihood obtained from the (\gls{MCMC}-)\gls{BPF} and (\gls{MCMC}-)\gls{FAAPF}. In this case, the \gls{MCMCFAAPF} outperforms the \gls{BPF} both in dimension $d=1$ and $d=5$.

Note that Assumptions~\ref{as:ergodicity}--\ref{as:lipschitz} and \ref{as:stability_mcmc_kernels}--\ref{as:stability_potential} are violated in this example. The results therefore appear to lend some support the conjecture that these assumptions are stronger than necessary for the results of Propositions~\mbox{\ref{prop:lr_inequality}--\ref{prop:clt}} to hold.

\begin{figure}[h]
 \centering
 \includegraphics[scale=1.2, trim = 2cm 0.3cm 1cm 1.5cm]{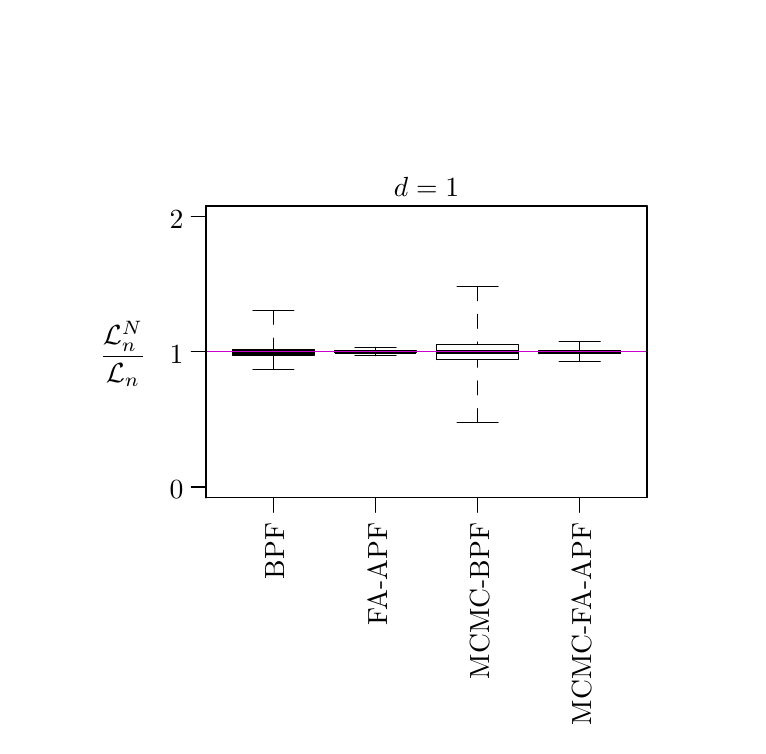}
 \includegraphics[scale=1.2, trim = 1.8cm 0.3cm 1.8cm 1.5cm]{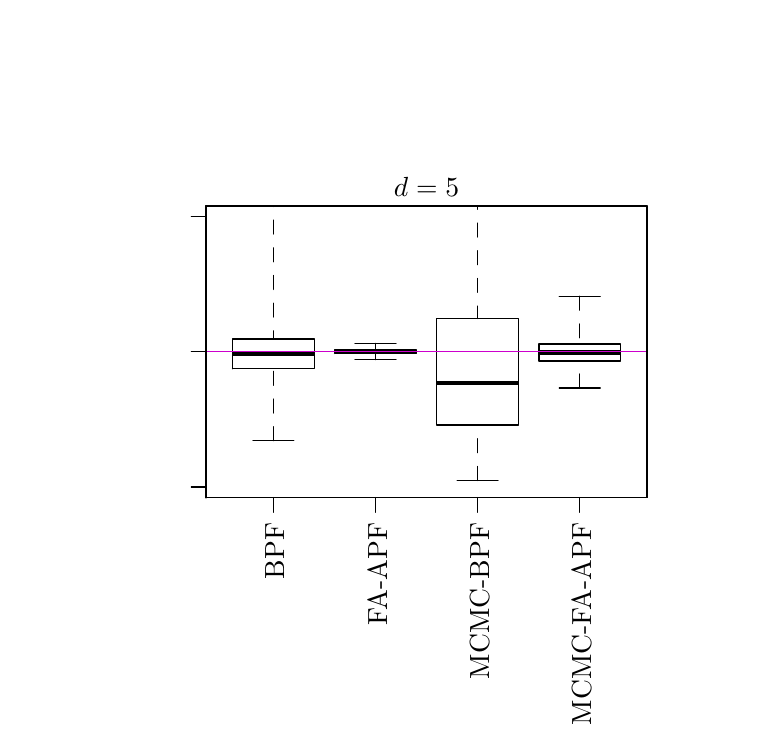}
 \caption{Relative estimates of the marginal likelihood $\calL_n$ in the linear Gaussian state-space model, generated by the algorithms discussed in Subsection~\ref{subsec:variance-variance_trade_off} using $N=10,000$ particles (with the \gls{MCMCFAAPF} using $N=10,000 - N_{\mathrm{burnin}}$ particles to compensate for the additional cost of generating the samples discarded as burn-in). Based on $1,000$ independent runs of each algorithm, each run using a different observation sequence of length $n=10$ sampled from the model. For the \gls{BPF} flow, $\calL_n = \normConst_{n+1} = \uTarget_n(\Potential_n)$ is estimated by $\calL_n^N \coloneqq \uTarget_n^N(\Potential_n)$; for the \gls{FAAPF} flow, $\calL_n = \normConst_n = \uTarget_n(\unitFun)$ is estimated by $\calL_n^N \coloneqq \normConst_n^N = \uTarget_n^N(\unitFun)$.}
 \label{fig:linear_ssm}
\end{figure}

\section{Conclusion}
\glsreset{PF}
\glsreset{MCMCPF}
\glsreset{FAAPF}
\glsreset{BPF}

In this work, we have established a \glsdesc{SLLN} and \glsdesc{CLT} for a class of algorithms known as sequential \gls{MCMC} methods or \glspl{MCMCPF} and provided conditions under which the associated errors can be controlled uniformly in time. When positive \gls{MCMC} operators are used within \glspl{MCMCPF}, the asymptotic variances of \gls{PF} estimators are always lower than the ones of the corresponding \gls{MCMCPF} estimators. However, even if the \gls{MCMC} kernels provide positively correlated draws, \glspl{MCMCPF} can remain of practical interest compared to \glspl{PF}. Indeed, there are many scenarios in which a sophisticated \gls{PF} such as the \gls{FAAPF} would significantly outperform a \gls{BPF} but cannot be implemented whereas the corresponding \gls{MCMCFAAPF} is essentially always applicable. If the \gls{MCMC} operators used within the \gls{MCMCFAAPF} are thus displaying ``reasonable'' \glsdesc{IACT}, the asymptotic variance of the resulting estimators can be smaller than the one of an implementable \gls{PF} such as the \gls{BPF}.

\section*{Acknowledgements}
This work was partially supported by funding from the Lloyd's Register Foun\-dation--Alan Turing Institute Programme on Data-Centric Engineering.

\renewcommand*{\bibfont}{\footnotesize}
\setlength{\bibsep}{3pt plus 0.3ex}
\bibliography{clt}

\appendix

\section{Martingale construction}
\label{app:martingale_construction}

In this section, we outline a number of useful martingale decompositions upon which the proofs of the $\mathbb{L}_r$-inequality in Proposition~\ref{prop:lr_inequality} and the  \gls{CLT} in Proposition~\ref{prop:clt} are based.

\subsection{Notation}

To simplify the notation, for any $n \geq 1$, we will hereafter often write
\begin{gather}
  \mcmcTarget{n}{N} \coloneqq \mcmcTarget{n}{\Target_{n-1}^N}, \quad \mcmcKern{n}{N} \coloneqq \mcmcKern{n}{\Target_{n-1}^N}, \quad \resolv{n}{N} \coloneqq \resolv{n}{\Target_{n-1}^N}, \quad \covarianceFunction{n}{N}  \coloneqq \covarianceFunction{n}{\Target_{n-1}^N},\\
  \mcmcTarget{n}{} \coloneqq \mcmcTarget{n}{\Target_{n-1}} (= \Target_n), \quad \mcmcKern{n}{} \coloneqq \mcmcKern{n}{\Target_{n-1}}, \quad \resolv{n}{} \coloneqq \resolv{n}{\Target_{n-1}}, \quad \covarianceFunction{n}{}  \coloneqq \covarianceFunction{n}{\Target_{n-1}}.
\end{gather}
Furthermore, we allow $\testFun \coloneqq (\testFun_n)_{n \geq 1}$ to denote a sequence of test functions where for any $n \geq 1$, $\testFun_n = (\testFun_n^u)_{1 \leq u \leq d} \in \boundedFunSet(\spacePath_n)^d$. For any $1 \leq u \leq d$, we also sometimes write $\testFun^u \coloneqq (\testFun_n^u)_{n \geq 1}$.

\subsection{Local errors}

We make use of the telescoping sum commonly used in the analysis of Feynman-Kac models, see \citep[see][Chapter~7]{delmoral2004feynman}:
\begin{align}
 \sqrt{N}[\Target_n^N- \Target_n](\testFun_n)
 & = \sqrt{N} \sum_{p=1}^n \mcmcTarget{p,n}{\Target_p^N}(\testFun_n) - \mcmcTarget{p,n}{\mcmcTarget{p}{N}}(\testFun_n) \label{eq:standard_telescoping-sum_decomposition}
\end{align}
with the convention that $\mcmcTarget{1}{\Target_{0}^N} = \Target_1$. Here, we have also defined $\mcmcTarget{p,n}{\mu}(\testFun_n) \coloneqq \mu Q_{p,n}(\testFun_n) /\mu Q_{p,n}(\unitFun)$. Key to the analysis are, therefore, the \emph{local errors}
\begin{align}
 \localErrorSingle_p^N(\testFun_p)
  \coloneqq \sqrt{N}[\Target_p^N - \mcmcTarget{p}{N}](\testFun_p)
  = [\widetilde{\localErrorSingle}_p^N+ R_p^N](\testFun_p),
\end{align}
where
\begin{align}
 R_p^N(\testFun_p)
 & \coloneqq \sqrt{N}[\Target_p^N - \tilde{\Target}_p^N](\testFun_p),\\
 \widetilde{\localErrorSingle}_p^N(\testFun_p)
 & \coloneqq \sqrt{N}[\tilde{\Target}_p^N - \mcmcTarget{p}{N}](\testFun_p), \label{eq:local_error_from_stationarity}
\end{align}
with $\tilde{\Target}_p^N(\testFun_p) \coloneqq \frac{1}{N}\sum_{i=1}^\nParticles \testFun_p(\ParticlePathStationary{p}{i})$. Here, for the purpose of facilitating the analysis, we have introduced the auxiliary Markov chain $\smash{(\ParticlePathStationary{p}{i})_{i \geq 1}}$ which evolves according to the same transition kernels as $\smash{(\ParticlePath{p}{i})_{i \geq 1}}$ but which is initialised from stationarity, \IE{} $\smash{\ParticlePathStationary{p}{1} \sim \mcmcTarget{p}{N}}$ and $\smash{\ParticlePathStationary{p}{i} \sim \mcmcKern{n}{N}(\ParticlePathStationary{p}{i-1}, \ccdot)}$, for $2 \leq i \leq N$. Note that $\smash{R_p^N(\testFun_p)}$ may therefore be viewed as the additional error introduced if the \gls{MCMC} chain is not initialised from stationarity at time~$p$. Tighter control of the errors could be obtained by explicitly coupling the two particle systems, but it is sufficient for our purposes to treat the two systems as being entirely independent. 

Using the tower property of conditional expectation it can be easily checked that $\E[\widetilde{\localErrorSingle}_p^N(\testFun_p)] = 0$ as in standard \glspl{PF}. However, contrary to standard \glspl{PF}, the particles $\smash{\ParticlePathStationary{p}{i}}$ and $\smash{\ParticlePathStationary{p}{j}}$, for $i \neq j$, are not necessarily conditionally independent given $\smash{\calF_{p-1}^{N,N}}$, where $\smash{\calF_{0}^{N,N} \coloneqq \{\emptyset, \Omega\}}$ and, for any $p \geq 1$ and $1 \leq k \leq N$,
\begin{equation}
 \calF_p^{N,k} \coloneqq \calF_{p-1}^{N,N} \vee \sigma(\ParticlePath{p}{i}, \ParticlePathStationary{p}{i} \mid  1 \leq i \leq k).  \label{eq:natural_filtration}
\end{equation}
Due to the lack of conditional independence, we obtain for any $1 \leq u \leq d$,
\begin{align}
 \!\!\!\!\!\!\!\E[\widetilde{\localErrorSingle}_p^N(\testFun_p^u)^2]
 & = \E\Bigl[\E\bigl(\widetilde{\localErrorSingle}_p^N(\testFun_p^u)^2\big|\calF_{p-1}^{N,N}\bigr)\Bigr] \label{eq:mcmc_finite-sample-variance}\\
 & = \E\Bigl[\mcmcTarget{p}{N}\bigl[(\testFun_p^u- \mcmcTarget{p}{N}(\testFun_p^u))^2\bigr]\Bigr]\\
 & \quad + 2 \sum_{j=1}^{N-1}\biggl(1-\frac{j}{N}\biggr) \E\Bigl[\mcmcTarget{p}{N}\bigl[(\testFun_p^u- \mcmcTarget{p}{N}(\testFun_p^u))(\mcmcKern{p}{N})^j(\testFun_p^u- \mcmcTarget{p}{N}(\testFun_p^u))\bigr]\Bigr].\!\!\!\!\!\!\!\! \label{eq:variance_local_error}
\end{align}
To see this, note that (conditional on $\smash{\calF_{p-1}^{N,N}}$) the inner expectation in \eqref{eq:mcmc_finite-sample-variance} is simply the variance of $\smash{ N^{-1/2} \sum_{i=1}^N h(\mathbf{Y}^i)}$, where $\smash{h(\mathbf{y}) \coloneqq \testFun_p^u(\mathbf{y}) - \mcmcTarget{p}{N}(\testFun_p^u)}$ and where $(\mathbf{Y}^i)_{i \geq 1}$ is a stationary Markov chain with invariant distribution $\mcmcTarget{p}{N}$ and transition kernels $\smash{\mcmcKern{p}{N}}$. The last line then follows by exploiting stationarity of that Markov chain and grouping equivalent terms. This is a generalisation, when $\mcmcKern{p}{\mu}$ is not perfectly mixing, of the result for a standard \gls{PF} in which $\smash{\mcmcKern{p}{\mu}(\path_p, \ccdot) = \mcmcTarget{p}{\mu} = \mcmcInit{p}{\mu}}$, for all $\path_p \in \spacePath_p$, in which case:
\begin{align}
 \E[\widetilde{\localErrorSingle}_p^N(\testFun_p^u)^2] = \E[\localErrorSingle_p^N(\testFun_p^u)^2] = \E[\mcmcTarget{p}{N}([\testFun_p^u- \mcmcTarget{p}{N}(\testFun_p^u)]^2)].
\end{align}

Following \citet{delmoral2010interacting, bercu2012fluctuations}, we further decompose \eqref{eq:local_error_from_stationarity} as
\begin{align}
 \widetilde{\localErrorSingle}_p^N(\testFun_p)
 & = \frac{1}{\sqrt{N}} \sum_{i=1}^N \bigl[\testFun_p(\ParticlePathStationary{p}{i}) - \mcmcTarget{p}{N}(\testFun_p)\bigr]\\
 & = \frac{1}{\sqrt{N}}\sum_{i=1}^{\smash{N}\vphantom{.}} \bigl[\resolv{p}{N}(\testFun_p)(\ParticlePathStationary{p}{i}) - \mcmcKern{p}{N}\resolv{p}{N}(\testFun_p)(\ParticlePathStationary{p}{i})\bigr] \quad \text{[by \eqref{eq:poisson_equation_1}]}\\
 & = \martingaleSingle_p^N(\testFun_p) + \remainderSingle_p^{N}(\testFun_p), \label{eq:martingale_single_decomposition}
\end{align}
where, letting $\smash{\ParticlePathStationary{p}{N+1}}$ be a random variable distributed, independently conditional upon $\smash{\calF_p^{N,N}}$, according to $\smash{\mcmcKern{p}{N}(\ParticlePathStationary{p}{N}, \ccdot)}$, we have (weakly) defined
\begin{align}
 \martingaleSingle_p^N(\testFun_p)
 & \coloneqq
 \frac{1}{\sqrt{N}}\sum_{i=1}^N \bigl[\resolv{p}{N}(\testFun_p)(\ParticlePathStationary{p}{i+1}) - \mcmcKern{p}{N}\resolv{p}{N}(\testFun_p)(\ParticlePathStationary{p}{i})\bigr],\\
 \remainderSingle_p^{N}(\testFun_p)
 & \coloneqq \frac{1}{\sqrt{N}} \bigl[\resolv{p}{N}(\testFun_p)(\ParticlePathStationary{p}{1}) - \resolv{p}{N}(\testFun_p)(\ParticlePathStationary{p}{N+1})\bigr].
\end{align}
This allows us to write the local error at time~$p$ as
\begin{align}
 \localErrorSingle_p^N
 & = \martingaleSingle_p^N + \remainderSingle_p^{N} + R_p^N.
 \label{eq:local_error_decomposition_full}
\end{align}

\subsection{Martingale approximation at time $p$}\label{subsec:martingaledecomposition}

For $1 \leq p \leq n$, let $\calF_p^N \coloneqq (\calF_p^{N,i})_{0 \leq i \leq N}$, where $\calF_p^{N,i}$ is defined as in \eqref{eq:natural_filtration} with the additional convention that $\smash{\calF_p^{N,0} = \calF_{p-1}^{N,N}}$. We now show that $\smash{\martingaleSingle_p^N(\testFun_p)}$ is a martingale while $\smash{\remainderSingle_p^{N}(\testFun_p)}$ and $\smash{R_p^N(\testFun_p)}$ are remainder terms which vanish almost surely as $N \to \infty$.
\begin{itemize}
  \item \textbf{Martingale.} For each $N\geq 1$, $\smash{\martingaleSingle_p^N(\testFun_p) \coloneqq \sum_{i=1}^N\martingaleDiffSingle_p^{N,i+1}(\testFun_p)}$ is the terminal value of a martingale (and these martingales form a triangular array) which is defined through the $\smash{\calF_p^N}$-martingale difference sequence $\smash{(\martingaleDiffSingle_p^{N,i+1}(\testFun_p))_{1 \leq i \leq N}}$, where
  \begin{align}
  \martingaleDiffSingle_p^{N,i+1}(\testFun_p)
  \coloneqq \frac{1}{\sqrt{N}}\bigl[\resolv{p}{N}(\testFun_p)(\ParticlePathStationary{p}{i+1}) - \mcmcKern{p}{N}\resolv{p}{N}(\testFun_p)(\ParticlePathStationary{p}{i})\bigr].
  \end{align}
  Indeed, for any $1 \leq i \leq N$ and any $1 \leq u,v \leq d$,
  \begin{align}
  \E\bigl[\martingaleDiffSingle_p^{N,i+1}(\testFun_p^u)\big|\calF_p^{N,i}\bigr] & = 0, \label{eq:martingale_expectation_time_p}\\
  \E\bigl[\martingaleDiffSingle_p^{N,i+1}(\testFun_p^u)\martingaleDiffSingle_p^{N,i+1}(\testFun_p^v)\big|\calF_p^{N,i}\bigr] & = \frac{1}{N}\covarianceFunction{p}{N}(\testFun_p^u, \testFun_p^v)(\ParticlePathStationary{p}{i}), \label{eq:martingale_covariance_time_p}
  \end{align}
  where the second line follows directly from the definition in \eqref{eq:definition_of_covFun}.

 \item \textbf{Remainder.}  By \eqref{eq:bound_on_resolvent}, for any $1 \leq u \leq d$, the remainder signed measure $\remainderSingle_p^{N}$ is bounded as
%
  \begin{align}
   \lvert \remainderSingle_p^{N}(\testFun_p^u) \rvert
   \leq \frac{2}{\sqrt{N}} \lvert \resolv{p}{N}(\testFun_p^u) \rvert
   \leq \frac{2 \boundResolvent{p} \lVert \testFun_p^u \rVert}{\sqrt{N}}. \label{eq:bound_on_first_remainder_measure}
  \end{align}


   By the same arguments as in \eqref{eq:bound_resolvent}, under Assumption~\ref{as:ergodicity}, for any $N, r \in \naturals$ and any $1 \leq u \leq d$, we have
  \begin{align}
    \E\bigl[ \lvert R_p^N(\testFun_p^u)\rvert^r \bigr]^{\frac{1}{r}} \leq \frac{\boundResolvent{p}\lVert \testFun_p^u \rVert}{\sqrt{N}}. \label{eq:bound_on_second_remainder_measure}
  \end{align}
  Note that by Markov's inequality and the Borel--Cantelli lemma, \eqref{eq:bound_on_second_remainder_measure} implies that $R_p^N(\testFun_p^u) \convergesAlmostSurely 0$ as $N \to \infty$.

\end{itemize}

\subsection{Martingale approximation up to time $n$}

The sum of all local errors up time~$n$ is given by
\begin{align}
 \localError_n^N(\testFun)
 & \coloneqq \sum_{p=1}^n \localErrorSingle_p^N(\testFun_p)
 = \martingale_n^N(\testFun) + \remainder_n^N(\testFun) + \calR_n^N(\testFun).
\end{align}
The three quantities appearing on the right hand side are defined as follows.
\begin{itemize}
  \item \textbf{Martingale.} Let $\smash{\calF^N \coloneqq (\calF_n^N)_{n \geq 1}}$ with $\smash{\calF_p^N \coloneqq \calF_p^{N,N}}$, then the terms
  \begin{align}
  \martingale_n^N(\testFun)
  \coloneqq \sum_{p=1}^n \martingaleSingle_p^N(\testFun_p), 
  \end{align}
  define an $\smash{\calF^N}$-martingale $\smash{(\martingale_n^N(\testFun))_{n \geq 1}}$. 

 \item \textbf{Remainder.}  Again,
 \begin{equation}
  \remainder_n^N(\testFun) \coloneqq \sum_{p=1}^n \remainderSingle_p^N(\testFun_p) \quad \text{and} \quad \calR_n^N(\testFun) \coloneqq \sum_{p=1}^n R_p^N(\testFun_p)
 \end{equation}
 constitute remainder terms. Note that for any $n \geq 1$ and any $1 \leq u \leq d$, by \eqref{eq:bound_on_first_remainder_measure} and \eqref{eq:bound_on_second_remainder_measure},
  \begin{equation}
   \lim_{N \to \infty} \lVert \remainder_n^N(\testFun^u) \rVert = 0 \quad \text{and} \quad \lvert \calR_n^N(\testFun^u)\rvert \convergesAlmostSurely 0.
  \end{equation}
\end{itemize}

\section{Convergence proofs}
\label{app:proofs}

\subsection{Auxiliary results needed for time-uniform bounds}

For any $1 \leq p \leq n$ and any $\testFun_n \in \boundedFunSet(\spacePath_n)$, we define
\begin{gather}
 r_{p,n}
 \coloneqq \sup_{\path_p, \pathAlt_p \in \spacePath_p} \frac{Q_{p,n}(\unitFun)(\path_p)}{Q_{p,n}(\unitFun)(\pathAlt_p)}, \quad 
 P_{p,n}(\testFun_n)
 \coloneqq \frac{Q_{p,n}(\testFun_n)}{Q_{p,n}(\unitFun)} \leq 1.
\end{gather}

In the remainder of this work, whenever we restrict our analysis to test functions $\testFun_n \in \smash{\boundedFunSet^\star(\spacePath_n)^d}$, we can replace $\beta(P_{p,n})$ by
\begin{align}
 \beta^\star(P_{p,n}) \coloneqq \sup\lvert P_{p,n}(\testFun_n')(\path_p) - P_{p,n}(\testFun_n')(\pathAlt_p)\rvert,
\end{align}
where the supremum is over all $\path_p, \pathAlt_p \in \spacePath_p$ and all $\testFun_n' \in \smash{\boundedFunSet^\star(\spacePath_n)}$ such that $\osc(\testFun_n') \leq 1$.

\begin{lemma}\label{lem:time-uniform_bounds}
 Under Assumptions~\ref{as:ergodicity} and  \ref{as:stability_mcmc_kernels}--\ref{as:stability_potential}, for any $1 \leq p \leq n$,
 \begin{align}
 \boundResolvent{p}
 & \leq  \boundResolventUniform, \quad \text{where} \quad \boundResolventUniform \coloneqq 2 \bar{\imath} / \varepsilon(K) < \infty,\\
 r_{p,n}
 & \leq \bar{r}, \quad \text{where} \quad \bar{r} \coloneqq \varepsilon(G)^{-(m+l)} \varepsilon(M)^{-1} < \infty,\\
 \beta^\star(P_{p,n})
 & \leq \bar{\beta}^{\lfloor (n-p)/m \rfloor}, \quad \text{where} \quad  \bar{\beta} \coloneqq (1 - \varepsilon(G)^{m+l} \varepsilon(\Mutation)^2) <1.
\end{align}
\end{lemma}
\begin{proof}
 This follows by similar arguments to those used in the proof of \citet[][Proposition~4.3.3]{delmoral2004feynman}. \hfill \ensuremath{_\Box}
\end{proof}

\subsection{Auxiliary results needed for the SLLN}

\begin{lemma}\label{lem:lr_inequality_for_martingale_u}
 Under Assumption~\ref{as:ergodicity}, for any $r \geq 1$, there exists $\boundExponent{r} < \infty$ such that for any $n \geq 1$, any $\testFun_n \in \boundedFunSet(\spacePath_n)$ and any $N \in \naturals$,
 \begin{equation}
   \E\bigl[ \lvert U_n^N(\testFun_n) \rvert^r \bigr]^{\frac{1}{r}} \leq 2 \boundExponent{r} \boundResolvent{n}  \lVert \testFun_n\rVert.
 \end{equation}
\end{lemma}
\begin{proof}
 Without loss of generality, assume that $\lVert \testFun_n \rVert \leq 1$. The quadratic variation associated with the martingale $U_n^N(\testFun_n)$ satisfies
  \begin{align}
  \sum_{\smash{i=1}}^N \E\bigl(\varDelta U_n^{N,i+1}(\testFun_n)^2\big|\calF_n^{N,i}\bigr)
  & = \tilde{\Target}_n^N \covarianceFunction{n}{N}(\testFun_n,\testFun_n) \quad \text{[by \eqref{eq:martingale_covariance_time_p}]}\\
  & \leq 4 \smash{\boundResolvent{n}^{2}}. \quad \text{[by \eqref{eq:boundedness_of_covFun}]}
  \label{eq:boundquadraticvariation}
 \end{align}
 Hence, by the Burkholder-Davis-Gundy inequality \citep[Theorem 17.7]{kallenberg2006foundations} there exists $\boundExponent{r} < \infty$ such that
 \begin{align}
  \E\bigl[ \lvert U_n^N(\testFun_n) \rvert^r \bigr]^{\frac{1}{r}} \leq 2 \boundExponent{r} \boundResolvent{n}.
 \end{align}
 This completes the proof. \hfill \ensuremath{_\Box}
\end{proof}

We are now ready to prove the $\mathbb{L}_r$-inequality in Proposition~\ref{prop:lr_inequality}.

\begin{proof}[of Proposition~\ref{prop:lr_inequality}]
 Without loss of generality, assume that $\lVert\testFun_n\rVert \leq 1$ for all $n \geq 1$ and that the constants $\boundExponent{r}$ in Lemma~\ref{lem:lr_inequality_for_martingale_u} satisfy $\inf_{r \geq 1} b_r \geq 1$. 
 
 We begin by proving the first part of the proposition, \IE{} the $\mathbb{L}_r$-error bound without the additional Assumptions~\ref{as:stability_mcmc_kernels}--\ref{as:stability_potential}. The proof proceeds by induction on $n$. At time~$n=1$, by Minkowski's inequality combined with Lemma~\ref{lem:lr_inequality_for_martingale_u} as well as \eqref{eq:bound_on_first_remainder_measure} and \eqref{eq:bound_on_second_remainder_measure}, we have
 \begin{align}
  \sqrt{N} \E\bigl[\lvert [\Target_1^N - \Target_1](\testFun_1) \rvert^r \bigr]^{\frac{1}{r}}
  & = \sqrt{N} \E\bigl[\lvert [\Target_1^N - \mcmcTarget{1}{N}](\testFun_1) \rvert^r \bigr]^{\frac{1}{r}}\\
  & \leq \E\bigl[\lvert U_1^N(\testFun_1) \rvert^r \bigr]^{\frac{1}{r}} + \E\bigl[\lvert L_1^N (\testFun_1) \rvert^r \bigr]^{\frac{1}{r}} + \E\bigl[\lvert R_1^N (\testFun_1) \rvert^r \bigr]^{\frac{1}{r}}\\
  & \leq 2 \boundExponent{r} \boundResolvent{1} + \frac{3 \boundResolvent{1}}{\sqrt{N}} \leq a_1  \boundExponent{r},
 \end{align}
 \EG\@ with $a_1 \coloneqq 5 \boundResolvent{1} < \infty$.
 
 Assume now that the first part of the proposition holds at time~$n-1$, for some $n > 1$. By Minkowski's inequality, 
 \begin{align}
  \MoveEqLeft \sqrt{N} \E\bigl[\lvert [\Target_n^N - \Target_n](\testFun_n) \rvert^r \bigr]^{\frac{1}{r}}\\
  & \leq \sqrt{N} \E\bigl[\lvert [\Target_n^N - \mcmcTarget{n}{N}](\testFun_n) \rvert^r \bigr]^{\frac{1}{r}} + \sqrt{N} \E\bigl[\lvert [\mcmcTarget{n}{N} - \Target_n](\testFun_n) \rvert^r \bigr]^{\frac{1}{r}} \label{eq:lr_error_induction_proof_bound:1}\\
  & \leq 2 \boundExponent{r} \boundResolvent{n} + \frac{3 \boundResolvent{n}}{\sqrt{N}} + \frac{2 \boundExponent{r} a_{n-1}}{\Target_{n-1}(\Potential_{n-1})} \label{eq:lr_error_induction_proof_bound:2}  \leq a_n \boundExponent{r},
 \end{align}
 \EG\@ with $a_n \coloneqq 5 \boundResolvent{n} + 2 a_{n-1} /\Target_{n-1}(\Potential_{n-1}) < \infty$. Here, the bound on the first term in \eqref{eq:lr_error_induction_proof_bound:1} follows by the same arguments as at time~$1$. The bound on the second term in \eqref{eq:lr_error_induction_proof_bound:1} follows from the following decomposition (note that $\Target_{n-1}(Q_n(\unitFun)) = \Target_{n-1}(\Potential_{n-1})$):
 \begin{align}
  \MoveEqLeft \Target_{n-1}(\Potential_{n-1}) \lvert [\mcmcTarget{n}{N} - \Target_n](\testFun_n) \rvert\\
  & = \bigl\lvert \Target_{n-1}(\Potential_{n-1}) \mcmcTarget{n}{N}(\testFun_n) - \Target_{n-1}^N(Q_n(\testFun_n)) + \Target_{n-1}^N(Q_n(\testFun_n)) - \Target_{n-1}(Q_n(\testFun_n))  \bigr\rvert\\
  & =  \bigl\lvert \mcmcTarget{n}{N}(\testFun_n) [\Target_{n-1} - \Target_{n-1}^N](Q_n(\unitFun)) + [\Target_{n-1}^N - \Target_{n-1}](Q_n(\testFun_n))  \bigr\rvert\\
  & \leq  \lVert \mcmcTarget{n}{N}(\testFun_n) \rVert  \lvert [\Target_{n-1} - \Target_{n-1}^N](Q_n(\unitFun))\rvert + \lvert [\Target_{n-1}^N - \Target_{n-1}](Q_n(\testFun_n)) \rvert.
 \end{align}
 Minkowski's inequality along with $\lVert Q_n(\unitFun)\rVert \leq 1$, $\lVert Q_n(\testFun_n) \rVert \leq 1$ and $\lVert \mcmcTarget{n}{N}(\testFun_n) \rVert \leq \lVert \testFun_n\rVert \leq 1$ combined with the induction assumption then readily yields the bound given in \eqref{eq:lr_error_induction_proof_bound:2}, \IE\@
 \begin{align}
  \sqrt{N} \E\bigl[\lvert [\mcmcTarget{n}{N} - \Target_n](\testFun_n) \rvert^r \bigr]^{\frac{1}{r}}
  \leq \frac{2 \boundExponent{r} a_{n-1}}{\Target_{n-1}(\Potential_{n-1})}.
 \end{align}
 This completes the first part of the proposition.
 
 As the bounds obtained through the previous induction proof cannot easily be made time-uniform, we prove the second part of the proposition via the more conventional telescoping-sum decomposition given in \eqref{eq:standard_telescoping-sum_decomposition}. Using the arguments in \citet[][pp.~244--246]{delmoral2004feynman}, we obtain the following bound for the $p$th term in the telescoping sum:
 \begin{align}
  \sqrt{N}\lvert [\mcmcTarget{p,n}{\Target_p^N} - \mcmcTarget{p,n}{\mcmcTarget{p}{N}}](\testFun_n) \rvert
  & \leq 2 \sqrt{N} \lvert [\Target_p^N - \mcmcTarget{p}{N}](\widebar{Q}_{p,n}^N(\testFun_n)) \rvert
  r_{p,n} 
  \beta(P_{p,n})\\
  & = 2 \lvert
  [
  U_p^N + L_p^N + R_p^N
  ]
  (\widebar{Q}_{p,n}^N(\testFun_n)) \rvert 
  r_{p,n}  
  \beta(P_{p,n}),
 \end{align}
where the second line is due to \eqref{eq:local_error_decomposition_full} and where
\begin{gather}
 \widebar{Q}_{p,n}^N(\testFun_n)
  \coloneqq \frac{Q_{p,n}^N(\testFun_n)}{\lVert Q_{p,n}^N(\testFun_n) \rVert},\\
 Q_{p,n}^N(\testFun_n)
  \coloneqq \frac{Q_{p,n}(\unitFun)}{\mcmcTarget{p}{N}(Q_{p,n}(\unitFun))} P_{p,n}\biggl(\testFun_n - \frac{\mcmcTarget{p}{N}(Q_{p,n}(\testFun_n))}{\mcmcTarget{p}{N}(Q_{p,n}(\unitFun))}\biggr).
\end{gather}
Hence, by \eqref{eq:standard_telescoping-sum_decomposition},
\begin{align}
  \MoveEqLeft \sqrt{N} \E\bigl[\lvert [\Target_n^N - \Target_n](\testFun_n) \rvert^r\bigr]^{\frac{1}{r}}\\
  & \leq 2 \sum_{p=1}^n \E\bigl[\lvert
  [
  U_p^N + L_p^N + R_p^N
  ]
  (\widebar{Q}_{p,n}^N(\testFun_n)) \rvert^r\bigr]^{\frac{1}{r}} r_{p,n} \beta(P_{p,n})\\
  & \leq 2 \sum_{p=1}^n \Bigl(2 \boundExponent{r} \boundResolvent{p} + \frac{3 \boundResolvent{p}}{\sqrt{N}}\Bigr) r_{p,n} \beta(P_{p,n}) \leq a_n b_r
 \end{align}
with $a_n \coloneqq 10 \sum_{p=1}^n \boundResolvent{p} r_{p,n} \beta(P_{p,n})$, where the last line follows from Minkowski's inequality combined with Lemma~\ref{lem:lr_inequality_for_martingale_u},  \eqref{eq:bound_on_first_remainder_measure} and \eqref{eq:bound_on_second_remainder_measure}. Since $\testFun \in \smash{\boundedFunSet^\star(\spacePath_n)}$, we can replace $\beta(P_{p,n})$ by $\beta^\star(P_{p,n})$ in the derivation above. Lemma~\ref{lem:time-uniform_bounds} then yields the time-uniform bound
\begin{align}
 a_n \leq 10 \boundResolventUniform \bar{r} \sum_{p=1}^n \bar{\beta}^{\lfloor (n -p)/m\rfloor}
 \leq 10 \boundResolventUniform \bar{r} m \sum_{n=0}^\infty \bar{\beta}^n \leq \frac{20 \bar{\imath} m }{\varepsilon(K) \varepsilon(M)^3 \varepsilon(G)^{2(m+l)}} \eqqcolon a.
\end{align}
This completes the proof. \hfill \ensuremath{_\Box}

\end{proof}

\subsection{Auxiliary results needed for the CLT}

\begin{lemma}\label{lem:convergence_of_covFun}
 Fix $n > 1$. For some $\mu \in \probMeasSet(\spacePath_{n-1})$ let $(\mu^N)_{N \geq 1}$ be a sequence of random probability measures on $(\spacePath_{n-1}, \sigFieldPath{n-1})$ such that $\mu^N(\testFun_{n-1}) \convergesAlmostSurely \mu(\testFun_{n-1})$ for all $\testFun_{n-1} \in \boundedFunSet(\spacePath_{n-1})$.

 Then under Assumptions~\ref{as:ergodicity} and \ref{as:lipschitz}, for all $(\testFun_n, \testFunAlt_n) \in \boundedFunSet(\spacePath_n)^2$,
 \begin{align}
  \lVert \covarianceFunction{n}{\mu^{\mathrlap{N}\,\,}}(\testFun_n, \testFunAlt_n) - \covarianceFunction{n}{\mu}(\testFun_n, \testFunAlt_n)\rVert
  \convergesAlmostSurely 0.
 \end{align}
\end{lemma}
\begin{proof}
 We use a similar argument to that used in the first part of the proof of \citet[Theorem~3.5]{bercu2012fluctuations}. That is, under Assumptions~\ref{as:ergodicity} and \ref{as:lipschitz}, and using \citet[Proposition~3.1]{bercu2012fluctuations}, a telescoping-sum decomposition allows us to find a constant $\smash{\boundIntegralOperatorLipschitzAlt{n} < \infty}$ and a family of bounded integral operator $(\integralOperatorLipschitzAlt{n}{\nu})_{\nu \in \probMeasSet(\spacePath_{n-1})}$ from $\boundedFunSet(\spacePath_{n-1})$ into $\boundedFunSet(\spacePath_{n})^2$ satisfying
 \begin{align}
  \sup_{\nu \in \probMeasSet(\spacePath_{n-1})}\int_{\boundedFunSet(\spacePath_{n-1})} \lVert h\rVert \integralOperatorLipschitzAlt{n}{\nu}((\testFun_n, \testFunAlt_n), \diff h) \leq \lVert \testFun_n\rVert \lVert \testFunAlt_n\rVert \boundIntegralOperatorLipschitzAlt{n}, \label{eq:bound_on_covFunDiff_1}
 \end{align}
 such that
 \begin{align}
  \lVert \covarianceFunction{n}{\mu^{\mathrlap{N}\,\,}}(\testFun_n, \testFunAlt_n) - \covarianceFunction{n}{\mu}(\testFun_n, \testFunAlt_n)\rVert
  & \leq \int_{\boundedFunSet(\spacePath_{n-1})} \lvert [\mu^N - \mu](h)\rvert \integralOperatorLipschitzAlt{n}{\mu}((\testFun_n,\testFunAlt_n), \diff h). \label{eq:bound_on_covFunDiff_2}
 \end{align}

 It remains to be shown that the \RHS{} in \eqref{eq:bound_on_covFunDiff_2} goes to zero almost surely. Let $\mathcal{D}$ denote the collection of bounded Borel\slash Borel-measurable functions from $\boundedFunSet(\spacePath_{n-1})$ to $\reals$ (with respect to the uniform and Euclidean norms, respectively). This set contains, among others, the mappings $h \mapsto \nu(h)$ induced by probability measures $\nu \in \probMeasSet(\spacePath_{n-1})$ via their action as linear integral operators. By Borel measurability of the norm, $\mathcal{D}$ also contains the function $h \mapsto \lVert h \rVert$. Since $\mu^N$ is a probability measure, we have $\lvert \mu^N(h) \rvert \leq \lVert h\rVert$, for any $h \in \boundedFunSet(\spacePath_{n-1})$, while \eqref{eq:bound_on_covFunDiff_1} ensures that $h \mapsto \lVert h \rVert$ is integrable. Hence, we can apply Lebesgue's dominated convergence theorem \citep[\EG][Theorem 1.21]{kallenberg2006foundations} to conclude that
 the \RHS{} of \eqref{eq:bound_on_covFunDiff_2} vanishes almost-surely as $N \to \infty$. \hfill \ensuremath{_\Box}
\end{proof}
%
%
%
%
%

We now prove the following proposition which adapts \citet[Proposition~4.3]{bercu2012fluctuations} \citep[see also][Theorem~9.3.1]{delmoral2004feynman} to our setting.

\begin{proposition} \label{prop:convergence_of_martingale_M}
 Let $\testFun \coloneqq (\testFun_n)_{n \geq 1}$, where $\testFun_n = (\testFun_n^u)_{1 \leq u \leq d} \in \boundedFunSet(\spacePath_n)^d$. The sequence of martingales $\martingale^N(\testFun) = (\martingale_n^N(\testFun))_{n \geq 1}$ converges in law as $N \to \infty$ to a Gaussian martingale $\martingale(\testFun) = (\martingale_n(\testFun))_{n \geq 1}$ such that for any $n \geq 1$ and any $1 \leq u,v\leq d$,
 \begin{align}
  \langle \martingale(\testFun^u), \martingale(\testFun^v)\rangle_n
  & = \sum_{p=1}^n \Target_p \covarianceFunction{p}{}(\testFun_p^u, \testFun_p^v).
 \end{align}
\end{proposition}
\begin{proof}
We begin by re-indexing the processes defined above. For any $(p,i) \in \naturals \times \{1,\dotsc,N\}$ with $1 \leq p$ and $1 \leq i \leq N$, define the bijection $\bijection^N$ by
\begin{align}
  \bijection^N(p,i) \coloneqq (p-1)N + i - 1,
\end{align}
Define the filtration $\calG^N \coloneqq (\calG_k^N)_{k \geq 1}$, where $\calG_k^N \coloneqq \vee_{(p,i)\colon \theta^N(p,i) \leq k} \calF_{p}^{N,i}$.
We then have $\martingale_n^N(\testFun) = \martingaleAlt_k^N(\testFun)$ whenever $k = \bijection^N(n,N)$, where
\begin{align}
  \martingaleAlt_k^N(\testFun) \coloneqq \sum_{j=1}^k \varDelta \martingaleAlt_j^N(\testFun),
\end{align}
defines an $\calG^N$-martingale $(\widetilde{\martingale}_k^N(\testFun))_{k \geq 1}$ with increments
\begin{align}
  \varDelta \martingaleAlt_j^N(\testFun)
  \coloneqq \martingaleDiffSingle_p^{N,i}(\testFun_p) \quad \text{for $\bijection^N(p,i) = j$.}
\end{align}
Indeed, by \eqref{eq:martingale_expectation_time_p} and \eqref{eq:martingale_covariance_time_p}, for any $1 \leq u,v\leq d$,
\begin{align}
  \E\bigl[\varDelta \martingaleAlt_j^N(\testFun^u)\big|\calG_{j-1}^N\bigr] & = 0, \label{eq:martingale_expectation_time_j}\\
  \E\bigl[\varDelta \martingaleAlt_j^N(\testFun^u)\varDelta \martingaleAlt_j^N(\testFun^v)\big|\calG_{j-1}^N\bigr] & = \frac{1}{N}\covarianceFunction{p}{N}(\testFun_p^u, \testFun_p^v)(\ParticlePathStationary{p}{i}), \label{eq:martingale_covariance_time_j}
\end{align}
where $p \geq 1$ and $1 \leq i \leq N$ satisfy $\bijection^N(p,i) = j$.

We now apply the \gls{CLT} for triangular arrays of martingale-difference sequences \citep[see, \EG,][Section~VII.8, Theorem~4; p.~543]{probability:theory:Shi95}. The Lindeberg condition is satisfied because the test functions $\testFun_n$ are bounded. Finally, for any $p \geq 1$,
\begin{align}
  \smashoperator{\sum_{\smash{k=(pN)+1}}^{(p+1)N}} \;\E[\varDelta \martingaleAlt_k^N(\testFun^u)\varDelta \martingaleAlt_k^N(\testFun^v)|\calG_{k-1}^N]
  & = \frac{1}{N} \smashoperator{\sum_{i=1}^{\smash{N}}} \covarianceFunction{p}{N}(\testFun_p^u, \testFun_p^v)(\ParticlePathStationary{p}{i})\\
  & = \tilde{\Target}_p^N \covarianceFunction{p}{N}(\testFun_p^u, \testFun_p^v)\\
  & \convergesAlmostSurely \Target_p \covarianceFunction{p}{}(\testFun_p^u, \testFun_p^v),
\end{align}
by Proposition~\ref{prop:slln} and Lemma~\ref{lem:convergence_of_covFun}, writing $\covarianceFunctionAlt{p}{N} \coloneqq \covarianceFunction{p}{N}(\testFun_p^u, \testFun_p^v)$ and $\covarianceFunctionAlt{p}{} \coloneqq \covarianceFunction{p}{}(\testFun_p^u, \testFun_p^v)$ to simplify the notation,
\begin{align}
 \MoveEqLeft\Prob\bigl(\bigl\{\lim\nolimits_{N \to \infty} \lvert \Target_p^N \covarianceFunctionAlt{p}{N} - \Target_p \covarianceFunctionAlt{p}{} \rvert = 0 \bigr\}\bigr)\\
 & \geq \Prob\bigl(\bigl\{\lim\nolimits_{N \to \infty} \lvert \Target_p^N(\covarianceFunctionAlt{p}{N} - \covarianceFunctionAlt{p}{}) \rvert + \lvert [\Target_p^N - \Target_p](\covarianceFunctionAlt{p}{})\rvert = 0 \bigr\}\bigr)\\
 & \geq \Prob\bigl(\bigl\{\lim\nolimits_{N \to \infty} \lVert \covarianceFunctionAlt{p}{N} - \covarianceFunctionAlt{p}{} \rVert + \lvert [\Target_p^N - \Target_p](\covarianceFunctionAlt{p}{})\rvert = 0 \bigr\}\bigr)\\
 & = 1.
\end{align}
As a result, for any $n \geq 1$, as $N \to \infty$,
\begin{equation}
 \smashoperator{\sum_{k=1}^{(n+1)N}} \E[\varDelta \martingaleAlt_k^N(\testFun^u)\varDelta \martingaleAlt_k^N(\testFun^v)|\calG_{k-1}^N]
 \convergesAlmostSurely \sum_{p=1}^n \Target_p \covarianceFunction{p}{}(\testFun_p^u, \testFun_p^v).
\end{equation}
This completes the proof. \hfill \ensuremath{_\Box}
\end{proof}

As an immediate consequence of Proposition~\ref{prop:convergence_of_martingale_M}, we obtain the following corollary. Its proof is a straightforward modification of the proof of \citet[Corollary~9.3.1]{delmoral2004feynman}.

\begin{corollary} \label{cor:convergence_of_local_errors}
 As $N \to \infty$, the sequence of random fields $\localErrorSingle^N = (\localErrorSingle_n^N)_{n\geq 1}$ converges in law (and in the sense of convergence of finite-dimensional marginals) to the sequence of independent and centred Gaussian random fields $\localErrorSingle = (\localErrorSingle_n)_{n \geq 1}$ with covariance function as defined in \eqref{eq:local_error_covariance_function}.  \hfill \ensuremath{_\Box}
\end{corollary}

We are now ready to prove the \gls{CLT}.

\begin{proof}[of Proposition~\ref{prop:clt}]
 The proof of the \gls{CLT} now follows by replacing \citet[Theorem~9.3.1 \& Corollary~9.3.1]{delmoral2004feynman} in the proofs of \citet[Propositions~9.4.1 \& 9.4.2]{delmoral2004feynman} with Proposition~\ref{prop:convergence_of_martingale_M} and Corollary~\ref{cor:convergence_of_local_errors}, respectively.

 For the time-uniform bound on the asymptotic variance in \eqref{enum:prop:clt:normalised}, we note that for any $1 \leq u \leq d$,
 \begin{align}
  \lVert  \widebar{Q}_{p,n}(\testFun_n^u - \Target_n(\testFun_n^u)) \rVert
  & = \biggl\lVert \frac{Q_{p,n}(\unitFun)}{\Target_p Q_{p,n}(\unitFun)} \bigl[P_{p,n}(\testFun_n^u) - \Psi_{p,n}^{\Target_p}P_{p,n}(\testFun_n^u)\bigr]\biggr\rVert\\
  & \leq r_{p,n} \lVert [\Id- \Psi_{p,n}^{\Target_p}] P_{p,n}(\testFun_n^u)\rVert\\
  & \leq 2 \lVert \testFun_n^u \rVert r_{p,n} \beta^\star(P_{p,n}),
 \end{align}
 where we have used that $\Target_n = \Psi_{p,n}^{\Target_p}P_{p,n}$, where
 \begin{align}
  \Psi_{p,n}^{\Target_p}(\diff \path_p) \coloneqq \frac{\Target_p(\diff \path_p)Q_{p,n}(\unitFun)(\path_p)}{\Target_p Q_{p,n}(\unitFun)}.
 \end{align}
 Hence, by \eqref{eq:boundedness_of_covFun} and Lemma~\ref{lem:time-uniform_bounds} the time-uniform bound on the asymptotic variance in \eqref{eq:bound_on_asymptotic_variance} holds, \EG{} with
\begin{align}
 c \coloneqq 8 \boundResolventUniform^2 \bar{r}^2 \sum_{p=1}^n \bar{\beta}^{2 \lfloor (n -p)/m\rfloor}
 \leq 8 \boundResolventUniform^2 \bar{r}^2 m \sum_{n=0}^\infty \bar{\beta}^n
 \leq \frac{32 \bar{\imath}^2 m }{\varepsilon(K)^2 \varepsilon(M)^4 \varepsilon(G)^{3(m+l)}}.
\end{align}
This completes the proof. \hfill \ensuremath{_\Box}
\end{proof}

\end{document}